\documentclass[conference]{IEEEtran}
\IEEEoverridecommandlockouts

\usepackage{balance}
\usepackage{hyperref}
\usepackage{comment}

\usepackage[T1]{fontenc}
\renewcommand{\sffamily}{\usefont{T1}{cmss}{m}{n}}

\newcommand{\DDegree}{{\small\textsf{DDegree}}\xspace}
\newcommand{\DDegCol}{{\small\textsf{DDegCol}}\xspace}
\newcommand{\SDegree}{{\small\textsf{SDegree}}\xspace}
\newcommand{\BitCol}{{\small\textsf{BitCol}}\xspace}
\newcommand{\DIST}{{\small\textsf{DIST}}\xspace}
\newcommand{\EBBkC}{{\small\textsf{EBBkC}}\xspace}

\usepackage{booktabs}
\usepackage{amsthm}
\usepackage{graphicx}

\makeatletter
  \def\th@definition{
  \thm@headfont{\itshape} 
  \thm@notefont{} 
  \rm
}
\makeatother

\theoremstyle{definition}
\newtheorem{definition}{Definition}
\newtheorem{lemma}{Lemma}
\newtheorem{theorem}{Theorem}

\newtheorem{example}{Example}

\usepackage{bm}
\usepackage{enumitem}

\usepackage[ruled, vlined, linesnumbered]{algorithm2e}

\newcommand{\DAG}{\vec{G}}
\newcommand{\V}{\V_{\DAG}}
\newcommand{\A}{\A_{\DAG}}
\newcommand{\Deg}[1]{d(#1)}
\newcommand{\Nbr}[1]{N_{\DAG}(#1)}
\newcommand{\VG}{\:\rule[-.7ex]{0.6pt}{2.5ex}\:V_{\DAG}\:\rule[-.7ex]{0.6pt}{2.5ex}\:}
\newcommand{\EG}{\:\rule[-.7ex]{0.6pt}{2.5ex}\:E_{\DAG}\:\rule[-.7ex]{0.6pt}{2.5ex}\:}

\newcommand{\identifying}{identifying\xspace}
\newcommand{\induced}[1]{$#1$-induced\xspace}
\newcommand{\prefix}[1]{$#1$-prefix\xspace}

\newcommand{\clink}[2]{L^{C}_{#1}(#2)\xspace}
\newcommand{\slink}[2]{L^{S}_{#1}(#2)\xspace}

\newcommand{\inc}[2]{\mathcal{I}(#1, #2)\xspace}
\newcommand{\exc}[2]{\mathcal{E}(#1, #2)\xspace}

\newcommand{\NIL}{\texttt{NIL}\xspace}

\usepackage[]{color}

\usepackage[]{subcaption}
\usepackage[]{graphicx}

\usepackage[]{xspace}

\usepackage{mathtools}
\newcommand{\redarc}[1]{\textcolor{red}{{\xrightarrow[]{\,{#1}\,}}}}
\newcommand{\bluearc}[1]{\textcolor{blue}{{\xrightarrow[]{\,{#1}\,}}}}

\usepackage{cite}
\usepackage{amsmath,amssymb,amsfonts}
\usepackage{algorithmic}
\usepackage{graphicx}
\usepackage{textcomp}
\usepackage{xcolor}
\def\BibTeX{{\rm B\kern-.05em{\sc i\kern-.025em b}\kern-.08em
    T\kern-.1667em\lower.7ex\hbox{E}\kern-.125emX}}
\begin{document}

\title{\textsc{DIST}: Efficient k-Clique Listing via\\ Induced Subgraph Trie}

\author{
\IEEEauthorblockN{Yehyun Nam, Jihoon Jang, Kunsoo Park}
\IEEEauthorblockA{
\textit{Seoul National University}\\
\{yhnam,jhjang,kpark\}@theory.snu.ac.kr
}
\and
\IEEEauthorblockN{Jianye Yang}
\IEEEauthorblockA{
\textit{Guangzhou University}\\
\hspace{0.5em}jyyang@gzhu.edu.cn
}
\and
\IEEEauthorblockN{Cheng Long}
\IEEEauthorblockA{
\textit{Nanyang Technological University}\\
c.long@ntu.edu.sg
}
}

\maketitle

\begin{abstract}
Listing $\bm{k}$-cliques plays a fundamental role in various data mining tasks, 
such as community detection and mining of cohesive substructures.
Existing algorithms for the $\bm{k}$-clique listing problem are built upon a general framework, which finds $\bm{k}$-cliques by recursively finding $\bm{(\!k\!-\!1\!)}$-cliques within subgraphs induced by the out-neighbors of each vertex.
However, this framework has inherent inefficiency of finding smaller cliques within certain subgraphs repeatedly.
In this paper, we propose an algorithm \DIST for the $\bm{k}$-clique listing problem.
In contrast to existing works, the main idea in our approach is to compute each clique in the given graph \emph{only once} and store it into a data structure called \emph{Induced Subgraph Trie},
which allows us to retrieve the cliques efficiently.
Furthermore, we propose a method to prune search space based on a novel concept called \emph{soft embedding of an $\bm{l}$-tree},
which further improves the running time.
We show the superiority of our approach in terms of time and space usage through comprehensive experiments conducted on real-world networks;
\DIST outperforms the state-of-the-art algorithm by up to two orders of magnitude in both single-threaded and parallel experiments.
\end{abstract}


\section{Introduction}
Graphs have been widely used to model relationships between entities in various domains,
such as social network analysis, bio-informatics, chemistry, and software engineering.
Mining cohesive subgraphs is a fundamental problem in analysis of real-world graphs,
such as community detection \cite{dourisboure2007extraction, kumar1999trawling, sozio2010community, chen2010dense, fang2022densest}, real-time story identification \cite{angel2012dense, letsios2016finding}, biological networks \cite{fratkin2006motifcut, saha2010dense}, and spam detection \cite{gibson2005discovering}.
Mining cohesive subgraphs is also used for building indexes for reachability and distance queries in graph databases \cite{cohen2003reachability, jin20093}, or compressing web graphs \cite{buehrer2008scalable}.
The most straightforward examples of cohesive subgraphs are cliques~\cite{clique-1, clique-2} and their variants, such as quasi-cliques~\cite{quasi-clique} and defective cliques~\cite{defective-clique-1, defective-clique-2}.
Additionally, there are many concepts of cohesive subgraphs based on cliques, including $k$-clique communities~\cite{palla2005uncovering}, $k$-clique densest subgraphs~\cite{tsourakakis2015k}, and nucleus decompositions~\cite{sariyuce2017nucleus}.

\begin{figure}
    \centering
    \includegraphics[scale=0.40, trim=0mm 14mm 0mm 10mm, clip=true]{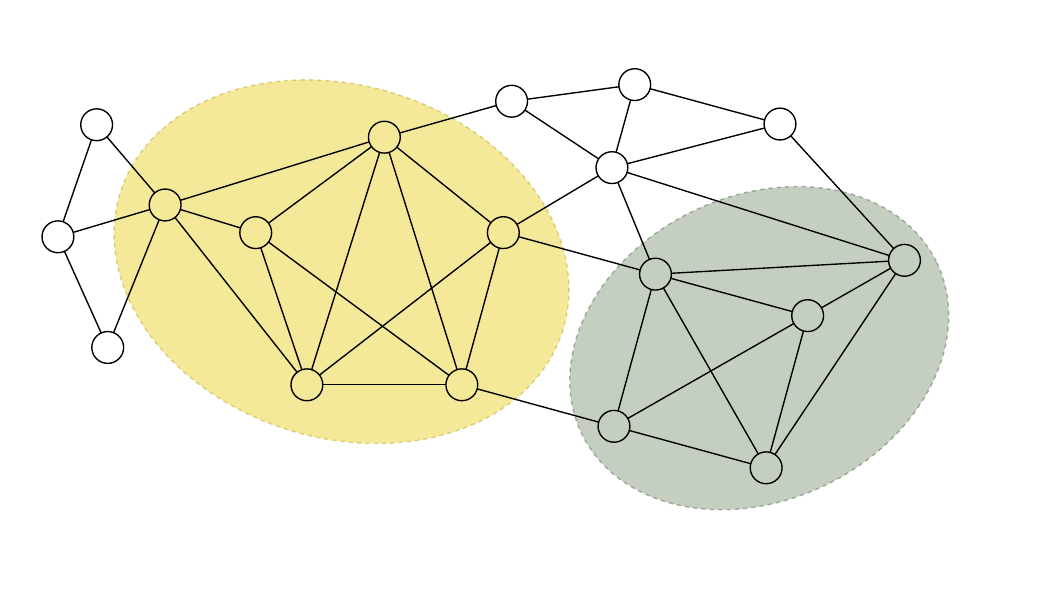}
    \caption{Two $k$-clique communities at $k\hspace{-0.1em}=\hspace{-0.1em}4$. A $k$-clique community is a union of $k$-cliques adjacent to each other, where adjacency means sharing $k\hspace{-0.1em}-\hspace{-0.1em}1$ vertices.}
    \label{fig:community}
\end{figure}

\begin{figure}[t]
    \centering
    \begin{subfigure}{0.49\linewidth}
        \includegraphics[scale=0.43, trim=11.5mm 0mm 12mm 4mm, clip=true]{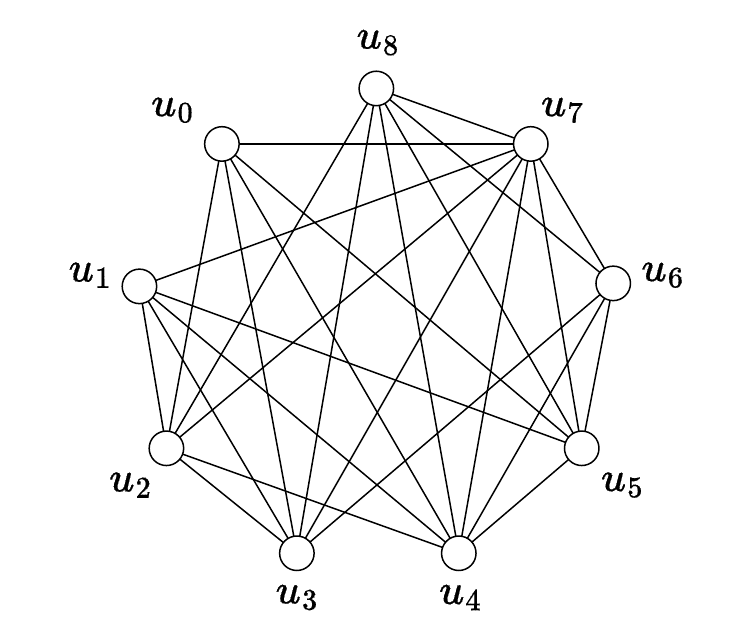}
        \caption{$G$}
        \label{fig:G}
    \end{subfigure}\hfill
    \begin{subfigure}{0.49\linewidth}
        \hspace*{2mm}
        \includegraphics[scale=0.43, trim=11.5mm 1mm 12mm 4mm, clip=true]{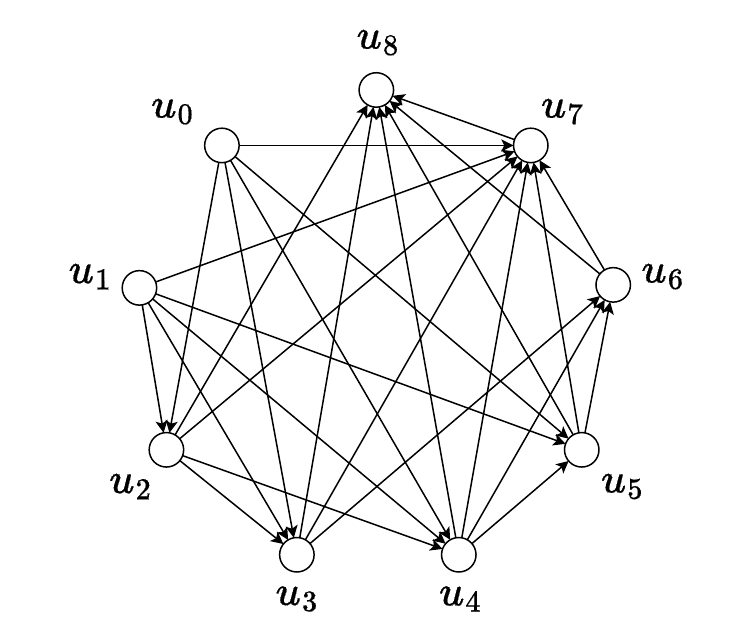}
        \caption{$\DAG$}
        \label{fig:DAG}
    \end{subfigure}
    \caption{Graph $G$ and DAG $\DAG$ based on degeneracy ordering.}
    \label{fig:graphs}
\end{figure}

\noindent
\textbf{Listing $\bm{k}$-cliques.} A set of vertices in a graph is called a \emph{clique} if there is an edge between every pair of distinct vertices,
that is, cliques are the densest subgraphs.
Listing all $k$-cliques (i.e., cliques of size $k$) plays an important role in various data mining tasks.

Algorithms for listing all $k$-cliques are used to detect communities within real-world networks \cite{hui2008human, palla2005uncovering, thilakarathna2013mobile}.
A $k$-clique community \cite{palla2005uncovering} is a union of $k$-cliques adjacent to each other, and algorithms for listing $k$-cliques are used to detect $k$-clique communities (see Figure~\ref{fig:community}).
Palla et al. \cite{palla2005uncovering} proposed a method to analyze statistical features of real-world networks based on $k$-clique communities.
Thilakarathna et al. \cite{thilakarathna2013mobile} proposed a content-replication algorithm to maximize the content dissemination within a limited number of replications, using $k$-clique community detection as a subroutine.
Sariy{\"u}ce et al. \cite{sariyuce2017nucleus} proposed a concept of nucleus decomposition to reveal a hierarchy of dense subgraphs, in which listing of all $k$-cliques is required.
Algorithms for $k$-clique listing are also used for the $k$-clique densest subgraph problem, introduced by Tsourakakis \cite{tsourakakis2015k} as a generalization of the densest subgraph problem.

Listing $k$-cliques is computationally hard; the decision version of the $k$-clique listing problem, i.e., to answer whether a graph includes a $k$-clique, is one of Karp's 21 NP-complete problems \cite{karp1972reducibility}.

\noindent
\textbf{Existing approaches and limitations.} The first practical algorithm for listing $k$-cliques was introduced by Chiba and Nishizeki \cite{chiba1985arboricity}.
Danisch et al. \cite{kClist} later improved this algorithm to operate on a directed version of the graph which is constructed based on a total ordering of vertices.
Their algorithm finds $k$-cliques by recursively finding $(k\hspace{-0.2em}-\hspace{-0.2em}1)$-cliques within subgraphs induced by the out-neighbors of each vertex.
This recursive process continues until it reaches the trivial task of finding $2$-cliques, i.e., edges.
Subsequent algorithms \cite{OrderingHeuristics, SetIntersectionSpeedup} have built upon this framework of Danisch et al.,
and recently, Wang et al. \cite{EBBkC} introduced an edge-oriented branching framework, which finds $k$-cliques by recursively finding $(k\!-\!2)$-cliques within subgraphs induced by the common out-neighbors of each edge.

However,
none have fully addressed the inherent inefficiency of this framework.
Specifically, the recursive nature of this framework leads to redundant computations, 
as the same smaller cliques within certain subgraphs are computed repeatedly.
Consider the task of finding 4-cliques in graph $\DAG$ of Figure~\ref{fig:DAG}.
To find all 4-cliques containing $u_0$, the recursive framework solves the subtask of finding all 3-cliques within the subgraph induced by $\{u_2, u_3, u_4, u_5, u_7\}$, i.e., the out-neighbors of $u_0$.
Since the out-neighbors of $u_1$ is the same as the out-neighbors of $u_0$, the same subtask arises when finding all 4-cliques containing $u_1$; the existing works based on this framework recompute all 3-cliques in the subgraph induced by $\{u_2, u_3, u_4, u_5, u_7\}$ from scratch.

\noindent
\textbf{Contributions.} 
We address the inherent inefficiency in the recursive framework for the $k$-clique listing problem, i.e., computing the same smaller cliques within certain subgraphs repeatedly.
In contrast to existing works, the main idea~in~our approach is to compute each clique in the given graph \emph{only once} and store it into a novel data structure called \emph{Induced Subgraph Trie},
which allows us to retrieve the cliques efficiently.
This algorithmic technique of storing items into a data structure and later retrieving them efficiently from the data structure is called \emph{memoization}~\cite{CLRS}, which has been used in the context of dynamic programming.
To the best of our~knowledge, our work is the first to utilize memoization for the~$k$-clique listing problem.
In dynamic programming,~applying~memoization is straightforward; simple values are stored into a table, and the values are retrieved from the table by a table lookup.
In $k$-clique listing, however, it is not obvious how to store the cliques in a compact way and how to retrieve them efficiently, since cliques in a graph may be intertwined with each other.

The contributions of our work are as follows.
\begin{itemize}[leftmargin=*]
    \item 
    We propose an innovative data structure called \emph{Induced Subgraph Trie} with additional attributes (called \emph{$l$-child links} and \emph{$l$-sibling links}), into which the cliques in the given graph $G$ are stored.
    By building a trie \cite{TAOCP3} with induced subgraphs of graph $G$ (nodes and black edges in Figure~\ref{fig:trie}) and then adding $l$-child links (red arcs in Figure~\ref{fig:trie}) and $l$-sibling links (blue arcs in Figure~\ref{fig:trie}), we can store the cliques of graph $G$ in a compact way.
    The cliques in $G$ are divided into two sets based on whether the minimum vertex in $G$ (with respect to a total ordering on $V_G$) is contained in a clique or not.
    By following $l$-child links (resp.~$l$-sibling links) in the Induced Subgraph Trie, we can efficiently retrieve the cliques that contain the minimum vertex (resp.~those that do not).
    \item 
    We introduce a novel concept called \emph{soft embedding of an $l$-tree}, which is a necessary condition for an existence of an $l$-clique. We propose a method to prune search space based on this concept.
    If our method finds that there is no soft embedding of an $l$-tree in a subgraph, we can safely prune the search space
    for finding $l$-cliques within the subgraph.
    Our pruning method further improves the running time, especially for large $k$ values.
    \item
    We show the superiority of our approach through comprehensive experiments conducted on 16 real-world networks.
    Experimental results show that
    our algorithm \DIST (\textsf{\textbf{D}ense \textbf{I}nduced \textbf{S}ubgraph \textbf{T}rie}) outperforms the state-of-the-art $k$-clique listing algorithm (\EBBkC in \cite{EBBkC})
    in both time and space usage.
    Specifically, \DIST solves $k$-clique listing faster than the state-of-the-art algorithm by up to two orders of magnitude in both single-threaded and parallel experiments.
\end{itemize}

\section{Preliminaries}
\subsection{Problem Statement}
\noindent
A \emph{graph} $G$ is a pair $(V_G, E_G)$ where $V_G$ is the set of vertices and $E_G\hspace{-0.1em} \subseteq\hspace{-0.1em} \big\{ \{u, v\}\hspace{-0.1em} \mid\hspace{-0.1em} u, v \in V_G \text{ and } u \neq v \big\}$ is the set of edges.
We say that vertices $u$ and $v$ in a graph $G$ are \emph{adjacent} if $\{u, v\} \in E_G$.
A \emph{clique} in a graph $G$ is a set of vertices $C \subseteq V_G$ such that every pair of distinct vertices in $C$ are adjacent in $G$.
The \emph{density} of a graph $G$ is defined as $|E_G|/\binom{|V_G|}{2}$.
A \emph{directed graph} $D$ is a pair $(V_{D}, E_{D})$ where $V_{D}$ is the set of vertices and $E_{D} \subseteq \big\{ (u, v) \mid u, v \in V_{D} \text{ and } u \neq v \big\}$ is the set of directed edges.
We say that a vertex $v$ is an \emph{out-neighbor} of a vertex $u$ if $(u, v) \in E_{D}$.
For a vertex $u$, we denote by $N_{D}(u)$ the set of out-neighbors of $u$ in $D$, and by $d_{D}(u)$ the number of out-neighbors of $u$ (i.e., the out-degree of $u$) in $D$.
For a directed graph $D$ and a set of vertices $S \subseteq V_D$, the subgraph of $D$ induced by $S$, denoted by $D[S]$, is a directed graph whose vertex set is $S$ and edge set consists of all the directed edges in $D$ with both endpoints in $S$.
A set of vertices $C$ of a directed graph $D$ is a clique if $C$ is a clique in the underlying undirected graph of $D$.

For a graph $G$ and a total ordering $\eta$ on $V_G$, the directed acyclic graph (DAG) of $G$ based on $\eta$ is a directed graph $\vec{G}$ such that 
$V_{\vec{G}}\hspace{-0.2em} =\hspace{-0.2em} V_G$ and $E_{\vec{G}}\hspace{-0.2em} =\hspace{-0.2em} \big\{ (u, v) \hspace{-0.2em}\mid\hspace{-0.2em} \{u, v\}\hspace{-0.1em} \in\hspace{-0.1em} E_G \mbox{\,and\,} u\hspace{-0.1em} \prec_{\eta}\hspace{-0.1em} v \big\}$.
Figure~\ref{fig:graphs} shows a graph $G$ and the DAG $\DAG$ based on the total ordering $u_0\hspace{-0.2em} \prec\hspace{-0.2em} u_1\hspace{-0.2em} \prec u_2\hspace{-0.2em} \prec\hspace{-0.2em} \cdots\hspace{-0.2em} \prec u_8$. 
We denote by $\DAG \setminus u$ the subgraph of $\DAG$ obtained by removing from $\DAG$ the vertex $u$ and its incident edges.
We denote by $\DAG_u$ the subgraph of $\DAG$ induced by the out-neighbors of a vertex $u$.
Throughout this paper, for a DAG $\DAG$ based on $\eta$, any subgraph of $\DAG$ is considered to be based on the restriction of $\eta$ to its vertex set.

Given a DAG $\DAG$ based on a total ordering $\eta$, the \emph{identifying string} of $\DAG$ is a string, denoted by $S_{\DAG}$,
which is obtained by concatenating the vertices of $\DAG$ in descending order with respect to $\eta$.
For example, in Figure~\ref{fig:DAG}, the identifying string of $\DAG_{u_1}$ is $u_7 u_5 u_4 u_3 u_2$.

The degeneracy \cite{kClist} of a graph $G$, denoted by $\delta(G)$, is the maximum integer $\delta$ such that there exists an induced~subgraph~of~$G$ with minimum degree $\delta$.
The degeneracy ordering~\cite{kClist} is a total ordering on the vertex set, obtained by removing a vertex with the minimum degree repeatedly until the graph becomes empty.

Table~\ref{tab:notations} lists the frequently used notations.
Whenever unambiguous, the notations are used without subscripts.
\begin{table}[b]
\setlength{\tabcolsep}{4pt}
  \caption{Frequently used notations.}
  \label{tab:notations}
  \small
  \begin{tabular}{cl}
    \toprule
    \textbf{Notation}      & \textbf{Meaning} \\
    \midrule
    $\delta(G)$                     & degeneracy of a graph $G$ \\*[0.5mm]
    $\omega(G)$                     & size of a maximum clique in a graph $G$ \\*[0.5mm]
    $d_{\DAG}(u)$ or $\Deg{u}$      & out-degree of a vertex $u$ in a DAG $\DAG$ \\*[0.5mm]
    $N_{\DAG}(u)$ or $N(u)$         & out-neighbors of a vertex $u$ in a DAG $\DAG$ \\*[0.5mm]
    $\DAG_u$                        & subgraph of a DAG $\DAG$ induced by $N_{\DAG}(u)$ \\*[0.5mm]
    $S_{\DAG}$                      & identifying string of a DAG $\DAG$ \\*[0.5mm]
    $\clink{t}{l}$ and $\slink{t}{l}$                  & $l$-child link and $l$-sibling link of a node $t$ \\*[0.5mm]
    $\DAG[t]$                       & subgraph of a DAG $\DAG$ corresponding a node $t$ \\
  \bottomrule
\end{tabular}
\end{table}
\begin{definition}[$k$-clique listing]
Given a graph $G$ and an integer $k$, the \emph{$k$-clique listing} problem is to list all the $k$-cliques in $G$.
\end{definition}
\noindent
For example, for the graph $G$ in Figure~\ref{fig:G}~and~$k~\hspace{-0.2em}=~\hspace{-0.2em}4$,~there~are $14$ $k$-cliques in $G$, i.e.,
$\{ u_0, u_2, u_3, u_7 \}$,
$\{ u_0, u_2, u_4, u_7 \}$, 
$\ldots$.

\noindent
\textbf{General framework.} Existing approaches~\cite{kClist, OrderingHeuristics, SetIntersectionSpeedup} to the $k$-clique listing problem follow the general framework proposed by Danisch et al., described in Algorithm~\ref{alg:framework}.
Given a graph $G$, the algorithm first computes a DAG $\DAG$ based~on~a total ordering of the vertex set (line~\ref{line:ordering}).
Then, it invokes \textsf{Listing}$(\DAG, k, \varnothing)$ (line~\ref{line:Listing}).
The function \textsf{Listing}$(\DAG, l, C)$, where~$\DAG$ is the subgraph of the original graph induced by the common neighbors of a clique $C$ and $l\hspace{-0.2em} =\hspace{-0.2em} k - |C|$,
is a recursive procedure to compute all the $l$-cliques within $\DAG$.
Combining the $l$-cliques with the clique $C$, it reports $k$-cliques in the original graph.
Throughout the procedure, the set $C$ is always a clique in the original graph.
In the base case, when $l\hspace{-0.1em}=\hspace{-0.1em}2$, the function \textsf{Listing}($\DAG, l, C$) iterates through the edge set $E_{\DAG}$ and reports each edge along with the set $C$ as a $k$-clique in the original graph (lines~4--6).
Otherwise, it invokes~a~recursive~call \textsf{Listing}($\DAG_u, l-1, C \cup \{u\})$ for each vertex $u$ in $\DAG$ (lines~7--9).
\subsection{Related Works}
\noindent
\textbf{$\bm{k}$-clique listing.}
Danisch et al. \cite{kClist} proposed the ordering-based framework for the $k$-clique listing problem, detailed in Section~II-A.
To construct a DAG of the input graph,~their algorithm \textsf{kClist} uses the degeneracy ordering.
They showed that their algorithm runs in $\mathcal{O}\big(k\cdot m\cdot (\delta(G)/2)^{k-2} \big)$ time, where $\delta(G)$ is the degeneracy of the input graph $G$.
They also proposed two strategies, namely \textsf{NodeParallel} and \textsf{EdgeParallel}, for parallelizing the ordering-based framework.
Li~et~al.~\cite{OrderingHeuristics} proposed two algorithms for the $k$-clique listing problem, namely \DDegree and \DDegCol, which use degree-ordering and color-ordering, respectively.
\DDegree initially generates a DAG based on the degeneracy ordering, but within subgraphs at the first level of the recursion, it reorders the vertices according to their degrees.
Similarly, \DDegCol applies greedy coloring within those subgraphs and reorders the vertices according to their color values.
The reordering enables pruning of the search space based on degree and color value, respectively.
Yuan~et~al.~\cite{SetIntersectionSpeedup} proposed two algorithms for the $k$-clique listing problem, namely \SDegree and \BitCol, along with preprocessing techniques for the input graph.
They focused mainly on speeding up set intersections for computing subgraphs induced by neighborhoods.
Recently, Wang et al. \cite{EBBkC} proposed an edge-oriented branching framework, in which the listing process invokes a recursive call for each edge, rather than for each vertex.
They also introduced an early termination technique to speed up the listing process within $t$-plexes \cite{plex}.

\noindent
\textbf{$\bm{k}$-clique counting.}
There has been a lot of research focused on counting or estimating the number of $k$-cliques.
Jain and Seshadhri \cite{jain2017fast} proposed a heuristic called \textsf{Tur{\'a}n-shadow}, which estimates the number of $k$-cliques based on Tur{\'a}n's theorem \cite{turan1941external}.
In a separate work, Jain and Seshadhri \cite{jain2020power} developed \textsf{Pivoter}, an algorithm that exactly counts $k$-cliques without actually enumerating all the cliques.
Recently, Ye et al. \cite{ye2022lightning} presented \textsf{DPColorPath}, a method that exactly counts the number of $k$-cliques in sparse regions of the graph while approximating the counts for dense regions of the graph via sampling. 
Note that methods for counting or estimating the number of $k$-cliques cannot be applied to listing $k$-cliques directly, since they either combinatorially count $k$-cliques without enumeration or samples only a portion of $k$-cliques.

\begin{algorithm}[t]
    \caption{The general framework for the $k$-clique listing problem \cite{kClist}}
    \label{alg:framework}
    \DontPrintSemicolon
    \small

    \SetKwBlock{Begin}{function}{end function}

    \SetKwFunction{Listing}{\textsf{Listing}}

    \SetKwInOut{Input}{Input}
    \SetKwInOut{Output}{Output}

    \Input{\enspace a graph $G$ and an integer $k$}
    \Output{\enspace all the $k$-cliques in $G$}
    \smallskip
    Compute a DAG $\vec{G}$ of $G$ based on total ordering of vertices.\;\label{line:ordering}
    Invoke \Listing{$\vec{G}$, $k$, $\varnothing$}.\;\label{line:Listing}
    \Begin(\Listing{$\vec{G}$, $l$, $C$}){
        \eIf{$l = 2$}{
            \ForEach{$(u, v) \in E_{\vec{G}}$}{
                Report $C \cup \{ u, v \}$ as a $k$-clique.\;\label{line:basecase}
            }
        }{
            \ForEach{$u \in V_{\vec{G}}$}{
                \Listing{$\vec{G}_u$, $l-1$, $C \cup \{ u \}$}\;\label{line:recursion}
            }
        }
    }
\end{algorithm}

\section{Overview of Algorithm}
\noindent
In this section, we present an overview of our algorithm for the $k$-clique listing problem (Algorithm~\ref{alg:listing}).
Our algorithm follows the general framework for the $k$-clique listing problem,
but it is different from the existing approaches~\cite{kClist, OrderingHeuristics, SetIntersectionSpeedup} in that
it finds every $l$-clique for $l \le k$ in the given graph $G$ only once,
while existing approaches may find an $l$-clique for $l < k$ many times.
Once our algorithm finds an $l$-clique for $l < k$, it memoizes (i.e., stores) the $l$-clique
so that it can be subsequently retrieved without repeating the same computation
(this technique is called \emph{memoization} in~\cite{CLRS}).
Since cliques within graph $G$ may be intertwined with each other, it is not obvious how to memoize the cliques in a compact way.
We propose a novel data structure, called \emph{Induced Subgraph Trie}, into which the cliques in graph $G$ are stored in a compact way and can be retrieved efficiently.
\begin{example}
Consider the task of listing all the $4$-cliques in the DAG $\DAG$ in Figure~\ref{fig:DAG}.
When we choose the vertex $u_0$ at depth 1 of the recursion, i.e., $C= \{ u_0 \}$, to report every $4$-cliques containing $u_0$, the remaining task is to compute all the $3$-cliques within the subgraph $\DAG\big[\{ u_2, u_3, u_4, u_5, u_7 \}\big]$.
The same task arises when we choose the vertex $u_1$ at the highest level of the recursion, i.e., $C = \{ u_1 \}$.
Instead of recomputing this task, our algorithm efficiently obtains all the $3$-cliques within the subgraph $\DAG \big[ \{ u_2, u_3, u_4, u_5, u_7 \} \big]$ using the Induced Subgraph Trie.
\end{example}
Additionally, we propose a pruning method based on a novel concept called \emph{soft embedding of an $l$-tree}, 
to reduce the search space of the Induced Subgraph Trie (Section~\ref{sec:pruning}).

Algorithm~\ref{alg:listing} shows \DIST, our solution for the $k$-clique listing problem.
\DIST differs from the general framework in that when the subgraph for the next recursive call is dense,
it invokes \textsf{ListingDense}, a routine for listing cliques in dense subgraphs using the Induced Subgraph Trie, instead of making another recursive call to itself (lines~12--15).
This strategy is taken because the computational cost of listing cliques within dense subgraphs is high, making it is more beneficial to memoize the cliques.
This distinct behavior leads to the name \textsf{\textbf{D}ense \textbf{I}nduced \textbf{S}ubgraph \textbf{T}rie}.
The criteria for determining whether we invoke \textsf{ListingDense} will be detailed in Section~\ref{sec:trie}.
\begin{algorithm}[t]
    \caption{\DIST:\hspace{0.2em}Our algorithm for listing $k$-cliques}
    \label{alg:listing}
    \DontPrintSemicolon
    \small

    \SetKwBlock{Begin}{function}{end function}

    \SetKwFunction{Listing}{\textsf{Listing}}
    \SetKwFunction{ListingDense}{\textsf{ListingDense}}

    \SetKw{continue}{continue}
    \SetKw{NIL}{NIL}
    
    \SetKwInOut{Input}{Input}
    \SetKwInOut{Output}{Output}

    \Input{a graph $G$ and an integer $k$}
    \Output{all the $k$-cliques in $G$}
    \smallskip
    Compute a DAG $\vec{G}$ of $G$ based on the degeneracy ordering.\;
    Let $T$ be an empty Induced Subgraph Trie.\;
    Invoke \Listing{$\vec{G}$, $k$, $\varnothing$, $T$}.\;
    \Begin(\Listing{$\vec{G}$, $l$, $C$, $T$}){
        \eIf{$l = 2$}{
            \ForEach{$(u, v) \in E_{\vec{G}}$}{
                Report $C \cup \{ u, v \}$ as a $k$-clique.
            }
        }{
            \ForEach{$u \in V_{\vec{G}}$}{
                \If{\rm $\DAG_u$ can be pruned}{
                    \continue
                }
                \eIf{\rm $\vec{G}_u$ is dense }{
                    \ListingDense{$\vec{G}_u$, $l-1$, $C \cup \{ u \}$, $T$, \NIL}\;
                }{
                    \Listing{$\vec{G}_u$, $l-1$, $C \cup \{ u \}$, $T$}\;
                }
            }
        }
    }
\end{algorithm}

\section{Induced Subgraph Trie}\label{sec:trie}
\noindent
In this section, we propose a novel data structure and an efficient algorithm for listing cliques in a dense region of a graph using the data structure.

\noindent
\textbf{The data structure.}
We first describe how to build a trie \cite{TAOCP3} with induced subgraphs of graph $G$ (nodes and black edges in Figure~\ref{fig:trie}) and then describe how to add $l$-child links (red arcs in Figure~\ref{fig:trie}) and $l$-sibling links (blue arcs in Figure~\ref{fig:trie}) to the trie.
We begin by introducing the concept of $d$-induced subgraphs.

\begin{definition}
    Let $\DAG$ be a DAG and $C$ be a $d$-clique in $\DAG$.
    An induced subgraph $\DAG'$ of $\DAG$ whose vertex set is the set of common neighbors of the vertices in $C$, i.e.,
    \begin{equation*}
        V_{\DAG'} = \bigcap_{u \in C} \Nbr{u}.
    \end{equation*}
    is called \emph{\induced{d}} with respect to $C$. Any (possibly non-proper) induced subgraph $\DAG''$ of $\DAG'$ such that $S_{\DAG''}$ is a prefix of $S_{\DAG'}$ is called \emph{\prefix{d}} with respect to $C$. 
\end{definition}

\begin{example}
In Figure~\ref{fig:DAG}, the subgraph of $\DAG$ induced by the set $\{ u_2, u_3, u_4, u_5, u_7 \}$ is a \induced{1} subgraph with respect to $\{ u_0 \}$.
Additionally, the subgraph of $\DAG$ induced by the set $\{ u_6, u_7, u_8 \}$ is a \induced{1} subgraph with respect to $\{ u_3 \}$ and a \induced{2} subgraph with respect to $\{ u_4, u_5 \}$ simultaneously.
\end{example}

\begin{figure}[t]
    \centering
    \includegraphics[width=\columnwidth, trim=0mm 3mm 0mm 3mm, clip=true]{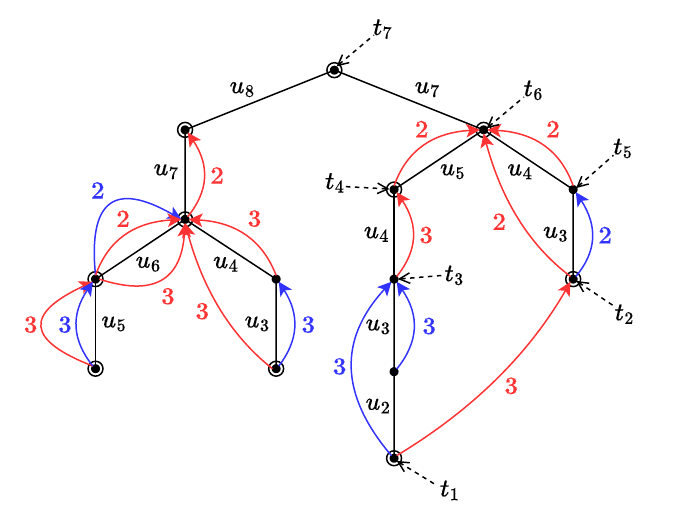}
    \caption{Induced Subgraph Trie for the DAG $\DAG$ in Figure~\ref{fig:DAG}. Red arc (resp. blue arc) labeled $l$ outgoing from a node $t$ is the $l$-child link (resp. $l$-sibling link) of $t$.}
    \label{fig:trie}
\end{figure}

\begin{table}[t]
    \centering
    \small
    \caption{Identifying string $S_{\DAG'}$ of \induced{d} subgraph $\DAG'$ and one $d$-clique  for each $d$ with respect to which $\DAG'$ is \induced{d}.}
    \begin{tabular}{c|c}
        \toprule
        $S_{\DAG'}$ & $d$-cliques \\
        \midrule
        $ u_{7} u_{5} u_{4} u_{3} u_{2} $   & $\{ u_{0} \} $ \\
        $ u_{8} u_{7} u_{4} u_{3} $         & $\{ u_{2} \} $ \\
        $ u_{8} u_{7} u_{6} $               & $\{ u_{3} \}, \{ u_{4}, u_{5} \} $ \\
        $ u_{8} u_{7} u_{6} u_{5} $         & $\{ u_{4} \} $ \\
        $ u_{8} u_{7} $                     & $\{ u_{6} \}, \{ u_{2}, u_{3} \}, \{ u_{4}, u_{5}, u_{6} \} $ \\
        $ u_{8} $                           & $\{ u_{7} \}, \{ u_{2}, u_{7} \}, \{ u_{2}, u_{3}, u_{7} \} $ \\
        $\varepsilon$                       & $\{ u_{8} \}, \{ u_{0}, u_{7} \}, \{ u_{0}, u_{2}, u_{7} \} $ \\
        $ u_{7} u_{4} u_{3} $               & $\{ u_{0}, u_{2} \} $ \\
        $ u_{7} $                           & $\{ u_{0}, u_{3} \}, \{ u_{0}, u_{2}, u_{3} \} $ \\
        $ u_{7} u_{5} $                     & $\{ u_{0}, u_{4} \} $ \\
        \bottomrule
    \end{tabular}
    \label{tab:strings}
\end{table}
Table~\ref{tab:strings} shows the identifying string of each subgraph $\DAG'$ that is \induced{d} for some $1\hspace{-0.2em}\le\hspace{-0.2em}d\hspace{-0.2em}<\hspace{-0.2em}4$,
along with $d$-cliques with respect to which $\DAG'$ is \induced{d}.
Note that the subgraph $\DAG\big[\{ u_2, u_3, u_4, u_5, u_7 \}\big]$ is \induced{1} with respect to both $\{ u_0 \}$ and $\{ u_1 \}$.
However, $\{ u_1 \}$ is omitted in Table~\ref{tab:strings} due to space limitations, along with some cliques for the other subgraphs. 

For a DAG $\DAG$ and an integer $k$, an \emph{Induced Subgraph Trie} is a trie~\cite{TAOCP3} (also known as prefix tree) for the set of \identifying strings of all the subgraphs that are \induced{d} with respect to $d$-cliques for $1\hspace{-0.2em}\le\hspace{-0.2em}d\hspace{-0.2em}<\hspace{-0.2em}k$.
For example, Figure~\ref{fig:trie} shows an Induced Subgraph Trie for the DAG $\DAG$ in Figure~\ref{fig:DAG} with $k\hspace{-0.1em}=\hspace{-0.1em}4$.
Each black edge (say, from node $s$ to node $t$) is labeled with a vertex $u$ in $\DAG$, which is also the label of the child node $t$, i.e., $u = \ell(t)$.
For example, the edge between $t_7$ and $t_6$ is labeled with vertex $u_7$, and $\ell(t_6) = u_7$.
Each node $t$ is associated with a string, denoted by $r(t)$, which is obtained by concatenating the labels on the path from the root to $t$.
For a node $t$, we denote by $\DAG[t]$ the induced subgraph of $\DAG$ whose identifying string is $r(t)$,
and we say that $t$ is the node corresponding to $\DAG[t]$.
Each leaf node corresponds to some \induced{d} subgraph of $\DAG$, while some internal nodes do not, i.e., they correspond to \prefix{d} subgraphs of $\DAG$.
A node that corresponds to a \induced{d} subgraph is circled.
Therefore, the tree in Figure~\ref{fig:trie} (nodes and black edges) is the trie for the set of identifying strings in Table~\ref{tab:strings}.

We now describe how to define $l$-child links and $l$-sibling links and how to add them to the trie.
Suppose we are to find $l$-cliques within a DAG $\DAG'$, where $l\hspace{-0.2em}=\hspace{-0.2em}k\hspace{-0.09em}-\hspace{-0.09em}d$ and $\DAG'$ is \induced{d} or \prefix{d} with respect to $C$.
Let $u$ be the minimum vertex in $\DAG'$,
i.e., the label of the node corresponding to $\DAG'$ in the Induced Subgraph Trie.
Then, the set of $l$-cliques in $\DAG'$ is divided into two sets as follows.
\begin{enumerate}[leftmargin=*]
    \item \emph{Inclusive set} $\mathcal{I}(\DAG', l)$ consists of the $l$-cliques in $\DAG'$ containing $u$, or equivalently,
    the $(l\hspace{-0.3em}-\hspace{-0.3em}1\hspace{-0.1mm})$-cliques in $\DAG'_u$ along with the vertex $u$.
    \item \emph{Exclusive set} $\mathcal{E}(\DAG', l)$ consists of the $l$-cliques in $\DAG'$ not containing $u$, or equivalently, the $l$-cliques in $\DAG' \setminus u$.
\end{enumerate}
\begin{example}
Consider the process of listing all the $4$-cliques in $\DAG$ in Figure~\ref{fig:DAG} that contain $\{ u_0 \}$.
Let $\DAG'\hspace{-0.1em}=\hspace{-0.1em}\DAG\big[\{ u_2, u_3, u_4,\allowbreak u_5, u_7 \} \big]$, which is \induced{1} with respect to $\{ u_0 \}$.
There are three $3$-cliques in $\DAG'$, namely $\{ u_2, u_3, u_7 \}$, $\{ u_2, u_4, u _7 \}$, and $\{ u_4, u_5, u_7 \}$.
The 3-cliques containing $u_2$, i.e., $\mathcal{I}(G', 3) = \big\{ \{ u_2, u_3, u_7 \}$, $\{ u_2, u_4, u _7 \} \big\}$, can be obtained by adding $u_2$ to each of $2$-cliques in $\DAG'_{u_2} = \DAG\big[\{ u_3, u_4,\allowbreak u_7 \}\big]$, namely, $\{ u_3, u_7 \}$ and $\{ u_4, u_7 \}$.
On the other hand, $\mathcal{E}(G', 3)=\big\{\{ u_4, u_5, u_7 \}\big\}$, where $\{ u_4, u_5, u_7 \}$ is the only $3$-clique in $\DAG\big[ \{ u_3, u_4,\allowbreak u_5, u_7 \} \big]$ (not containing $u_2$).
\end{example}
This process can be done recursively to compute each of the inclusive sets and the exclusive sets.
The Induced Subgraph Trie contains all the nodes that correspond to these subgraphs $\DAG'_u$ and $\DAG' \setminus u$;
\begin{enumerate}[leftmargin=*]
    \item $\DAG'_u$ is a \induced{(d\hspace{-0.2em}+\hspace{-0.2em}1)} subgraph with respect to the set $C \cup \{ u \}$ since $\DAG'$ is a \induced{d} or \prefix{d} subgraph, and
    \item $\DAG' \setminus u$ is a \induced{d} or \prefix{d} subgraph with respect to the set $C$ since $S_{\DAG' \setminus u}$ is a prefix of $S_{\DAG'}$, given that $u$ is the minimum vertex in $\DAG'$.
\end{enumerate}
When listing the cliques within some subgraph for the first time, we find them using the previous recursive process.
During the computation, we memoize the cliques within the Induced Subgraph Trie,
so that
they can be retrieved without  repeating the same computation.
To retrieve the inclusive set and the exclusive set efficiently, we introduce
two types of links to the Induced Subgraph Trie: \emph{$l$-child links} are used for the inclusive set and \emph{$l$-sibling links} are used for the exclusive set.
Cliques are compactly stored in the Induced Subgraph Trie using $l$-sibling links and $l$-child links.

Let $d$ and $l$ be integers such that $1 \le d < k$ and $l = k - d$, and $t$ be a non-root node such that 
$\DAG[t]$ is a \induced{d} or \prefix{d} subgraph with respect to a $d$-clique $C$.
Let $\DAG' = \DAG[t]$ and $u = \ell(t)$, i.e., the minimum vertex in $\DAG'$.
\begin{definition}
The $l$-child link of a node $t$, denoted by $\clink{t}{l}$, is a link towards another node $t_c$,
where $t_c$ is the node on the path from the root to the node corresponding to $\DAG'_u$ such that $(l-1)$-cliques in $\DAG[t_c]$ along with the vertex $u$ forms the inclusive set $\mathcal{I}(\DAG', l)$.
If there are multiple such nodes, $t_c$ is the one nearest to the root among them.
\end{definition}
\begin{definition}
The $l$-sibling link of a node $t$, denoted by $\slink{t}{l}$, is a link towards another node $t_s$,
where $t_s$ is the node on the path from the root to the node corresponding to $\DAG' \setminus u$ such that $\DAG'[t_s]$ includes the exclusive set $\mathcal{E}(\DAG', l)$.
If there are multiple such nodes, $t_s$ is the one nearest to the root among them.
\end{definition}
In Figure~\ref{fig:trie}, the red arc (resp. blue arc) labeled $l$ outgoing from a node $t$ represents the $l$-child link $L^{C}_{t}(l)$ (resp. $l$-sibling link $L^{S}_{t}(l)$).
Note that the $1$-child link of any node is a link towards the root and the $1$-sibling link of any node is a link towards its parent node.
For simplicity, $1$-child links, $1$-sibling links, and links towards the root are omitted in Figure~\ref{fig:trie}.
\begin{example}\label{example:links}
Consider the task of listing all the $4$-cliques in $\DAG$ in Figure~\ref{fig:DAG} that includes $\{ u_0 \}$, when they are already memoized in the Induced Subgraph Trie.
This can be done by retrieving all the $3$-cliques in the subgraph $\DAG' = \DAG\big[\{ u_2, u_3, u_4, u_5, u_7 \}\big]$, which is the \induced{1} subgraph with respect to $\{ u_0 \}$.
In Figure~\ref{fig:trie}, the node $t_1$ corresponds to the subgraph $\DAG'$.
There is a $3$-child link $\clink{t_1}{3}$ (i.e., a red arc labeled $3$) from $t_1$ to node $t_2$.
Note that $\DAG[t_2]$ includes all the $2$-cliques in $\DAG'_{u_2}$, i.e., $\big\{ \{ u_3, u_7 \}, \{ u_4, u_7 \} \big\}$, but $\DAG[t_5]$ does not.
Also, there is a $3$-sibling link $\slink{t_1}{3}$ (i.e., a blue arc labeled $3$) from $t_1$ to node $t_3$.
Again, $\DAG[t_3]$ includes all the $3$-cliques in $\DAG' \setminus u_2$, i.e., $\big\{ \{ u_4, u_5, u_7 \} \big\}$, but $\DAG[t_4]$ does not.

The retrieval of all the $l$-cliques within a \induced{d} or \prefix{d} subgraph $\DAG[t]$ begins by following either the $l$-child link or the $l$-sibling link of $t$.
For the inclusive set $\inc{\DAG[t]}{l}$, the retrieval begins by following the $l$-child link of $t$, i.e., $\clink{t}{l}$.
To retrieve the remaining $l-1$ vertices of the $l$-cliques, it proceeds recursively by following the $(l-1)$-child/sibling links of $\clink{t}{l}$.
In this example, the retrieval of $\inc{\DAG[t_1]}{3} = \big\{ \{ u_2, u_3, u_7\}, \{ u_2, u_4, u_7\} \big\}$ begins by following the $3$-child link $\clink{t_1}{3} = t_2$;
the retrieval of $\{ u_2, u_3, u_7 \}$ proceeds as \textcolor{black}{$t_1 \redarc{3} t_2 \redarc{2} t_6 \redarc{1} t_7$},
and the retrieval of $\{ u_2, u_4, u_7 \}$ as \textcolor{black}{$t_1 \redarc{3} t_2 \bluearc{2} t_5 \redarc{2} t_6 \redarc{1} t_7$}.

The retrieval of the exclusive set $\exc{\DAG[t]}{l}$ starts by following the $l$-sibling link of $t$, i.e., $\slink{t}{l}$.
Since $\exc{\DAG[t]}{l}$ is the same as the set of $l$-cliques in $\DAG[\slink{t}{l}]$, the retrieval proceeds by following the $l$-child/sibling link of $\slink{t}{l}$.
In this example, the retrieval of $\exc{\DAG[t_1]}{3} = \big\{ \{ u_4, u_5, u_7 \} \big\}$ begins by following the $3$-sibling link $\slink{t_1}{3}=t_3$; the retrieval of $\{ u_4, u_5, u_7 \}$ proceeds as \textcolor{black}{$t_1 \bluearc{3} t_3 \redarc{3} t_4 \redarc{2} t_6 \redarc{1} t_7$}.
\end{example}
\begin{algorithm}[t]
    \caption{Listing $k$-cliques for dense subgraphs}
    \label{alg:listingDense}
    \DontPrintSemicolon
    \small
    
    \SetKwInOut{Input}{Input}
    \SetKwInOut{Output}{Output}

    \SetKwBlock{Begin}{function}{end function}

    \SetKwFunction{IsPrunable}{IsPrunable}
    \SetKwFunction{Density}{Density}
    \SetKwFunction{ListingDense}{\textsf{ListingDense}}
    \SetKwFunction{ListingIST}{\textsf{ListingIST}}
    \SetKwFunction{IsPrunable}{IsPrunable}

    \SetKw{continue}{continue}
    
    \Input{a \induced{d} subgraph $\DAG$ with respect to a clique~$C$, an integer $l\hspace{-0.1em}=\hspace{-0.1em}k\hspace{-0.1em}-\hspace{-0.1em}|C|$, an Induced Subgraph Trie $T$}
    \Output{$k$-cliques in the original graph formed by combining the $l$-cliques in $\DAG$ with the set $C$}
    
    \smallskip

    \Begin(\ListingDense{$\DAG$, $l$, $C$, $T$} ){
        Insert the \identifying string of $\DAG$ into $T$ and let $\langle t_0, t_1, t_2, ..., t_{|V_{\DAG}|} \rangle$ be the path from the root to the node corresponding to $\DAG$.\;
        \eIf{$l = 2$}{
            \ForEach{$(u, v) \in E_{\vec{G}}$}{
                Report $C \cup \{ u, v \}$ as a $k$-clique.\;
            }
            \For{\rm $i := 1$ to $|V_{\DAG}|$}{
                $u := \ell(t_i)$\;
                \If{\rm $t_i.\textsf{Marked}[l]$ is false}{
                    Set $\slink{t_i}{l}$.\label{line:update-sibling2}\;
                    \eIf{$\Nbr{u} = \varnothing$}{
                        $\clink{t_i}{l} := $ the root node of $T$\;
                    }{
                        Insert\hspace{0.28em}the\hspace{0.28em}identifying\hspace{0.28em}string\hspace{0.28em}of\hspace{0.28em}$\DAG_{u}$\hspace{0.28em}into\hspace{0.28em}$T$\hspace{0.28em}and\hspace{0.28em}let\hspace{0.28em}$t^*$\hspace{0.28em}be\hspace{0.28em}the\hspace{0.28em}corresponding\hspace{0.28em}node.\;
                        $\clink{t_i}{l} := t^{*}$\label{line:update-child2}\;
                    }
                    $t_i.\textsf{Marked}[l] := \text{true}$\;
                }
            }
        }{
            \For{\rm $i := 1$ to $|V_{\DAG}|$}{
                $u := \ell(t_i)$\;
                \eIf{\rm $t_i.\textsf{Marked}[l]$ is true}{
                    \ListingIST{$\clink{t_i}{l}$, $l\hspace{-0.1em}-\hspace{-0.1em}1$, $C \cup \{ u \}$}\;\label{line:ListingIST}
                }{
                    Set $\slink{t_i}{l}$.\label{line:update-sibling}\;
                    \eIf{\rm $\DAG_u$ can be pruned}{
                        $\clink{t_i}{l} := $ the root node of $T$\;
                    }{
                        \ListingDense{$\DAG_{u}$, $l\hspace{-0.1em}-\hspace{-0.1em}1$, $C \cup \{ u \}$, $T$}\;\label{line:ListingDense}
                        Set $\clink{t_i}{l}$.\label{line:update-child}\;
                    }
                    $t_i.\textsf{Marked}[l] := \text{true}$\;\label{line:mark}
                }
            }
        }
    }
\end{algorithm}

\noindent
\textbf{The algorithm.}
In our algorithm \DIST (Algorithm~\ref{alg:listing}), \textsf{Listing} invokes \textsf{ListingDense} instead of another recursive call when
$l\hspace{-0.1em}-\hspace{-0.1em}1\hspace{-0.1em}>\hspace{-0.1em}3$ (i.e., the upcoming clique listing task does not degenerate to listing triangles or edges)
and the following criterion is satisfied:
\begin{equation}
    \text{(density of $\DAG_u$)} > \frac{l-2}{l-1} \cdot \frac{\big|V_{\DAG_u}\big|}{\big|V_{\DAG_u}\big|-1}
\end{equation}
This criterion is derived from Tur{\'a}n's theorem \cite{turan1941external},
which states that any graph with $n$ vertices can have at most
$\frac{r-1}{r}\cdot\frac{n^2}{2}$ edges if the graph does not include any $r$-clique.
If the density of the subgraph $\DAG_u$ is greater than $\frac{l-2}{l-1} \cdot \frac{\rule[-1ex]{0.5pt}{2.2ex}V_{\DAG_u}\rule[-1ex]{0.5pt}{2.2ex}}{\rule[-1ex]{0.5pt}{2.2ex}V_{\DAG_u}\rule[-1ex]{0.5pt}{2.2ex}\;-1}$, 
then it has more than $\frac{l-2}{l-1} \cdot \frac{\rule[-1ex]{0.5pt}{2.2ex}V_{\DAG_u}\rule[-1ex]{0.5pt}{2.2ex}\,^2}{2}$ edges,
and we are guaranteed to find an $(l\hspace{-0.1em}-\hspace{-0.2em}1\hspace{-0.1em})$-clique within $\DAG_u$ according to Tur{\'a}n's theorem.
On the other hand, if $\DAG_u$ is sparse, we have no guarantee of finding an $(l\hspace{-0.2em}-\hspace{-0.2em}1\hspace{-0.1em})$-clique within $\DAG_u$, and even if there is one, the computational cost of the listing process is low, and it is not worth creating the Induced Subgraph Trie in this case.

The routine \textsf{ListingDense} in Algorithm~\ref{alg:listingDense} is dedicated to listing cliques in dense subgraphs.
\textsf{ListingDense} finds the same $l$-cliques as \textsf{Listing} does.
However, when the cliques have already been computed and memoized in the Induced Subgraph Trie $T$, it efficiently retrieves these cliques using $T$.
\begin{definition}
A non-root node $t$ in the Induced Subgraph Trie $T$ is \emph{$l$-memoized} if and only if
\begin{itemize}[leftmargin=*]
    \item $\clink{t}{l}$ and $\slink{t}{l}$ are created, and
    \item $\slink{t}{l}$ is $l$-memoized and $\clink{t}{l}$ is $(l-1)$-memoized.
\end{itemize}
We assume that the root node is $l$-memoized for any $l < k$.
\end{definition}\noindent
If a node $t$ is $l$-memoized, then $l$-cliques within $\DAG[t]$ can be retrieved through a process described in Example~4.
Suppose \textsf{ListingDense} is invoked with respect to a DAG $\DAG$ and an integer $l$.
To track whether the $l$-cliques in $\DAG$ are memoized, the identifying string of $\DAG$ is inserted into the Induced Subgraph Trie and each node $t$ of the Induced Subgraph Trie has an attribute $t.\textsf{Marked}[l]$.
This attribute $t.\textsf{Marked}[l]$ is initially set to false, and \textsf{ListingDense} sets $t.\textsf{Marked}[l]$ to true if and only if $t$ is $l$-memoized.

We now explain in detail how to create $l$-sibling links and $l$-child links, and how to set \textsf{Marked}.
Let $\DAG$ be a DAG and $\langle t_0, t_1, \ldots, t_{|V_{\DAG}|} \rangle$ be the path from the root to the node corresponding to $\DAG$ in the Induced Subgraph Trie, where
$t_0$ is the root of the Induced Subgraph Trie.
Consider the for-loop at lines $17$--$28$ in Algorithm~\ref{alg:listingDense}.
Since the identifying string of $\DAG$ is obtained by concatenating the vertices of $\DAG$ in descending order,
we have $\ell(t_1)\hspace{-0.12em}\succ\hspace{-0.12em}\ell(t_2)\hspace{-0.12em}\succ\hspace{-0.12em}\cdots\hspace{-0.12em}\succ\hspace{-0.12em}\ell\big(t_{|V_{\DAG}|}\big)$.
At the $i$-th iteration, we list the $l$-cliques in $\DAG[t_i]$ that contains $\ell(t_i)$, 
i.e., the set $\mathcal{I}(\DAG[t_i], l)$, 
either from the Induced Subgraph Trie (line~\ref{line:ListingIST}) or through a recursive call (line~\ref{line:ListingDense}).
Prior to the $i$-th iteration, 
the sets $\mathcal{I}(\DAG[t_1], l), \mathcal{I}(\DAG[t_2], l), \ldots, \mathcal{I}(\DAG[t_{i-1}], l)$ are listed, 
and they collectively form the set $\mathcal{E}(\DAG[t_i], l)$.
At line~\ref{line:update-sibling}, we must set $\slink{t_i}{l}$ to a node $t_j$ on the path $\langle t_0, t_1, \ldots, t_{i-1} \rangle$, where $j$ is the smallest index such that $\DAG[t_j]$ includes all of $\exc{\DAG[t_i]}{l}$, 
or equivalently, $j$ is the smallest index such that $\mathcal{I}(\DAG[t_j], l)$ is not empty but all of $\inc{\DAG[t_{j+1}]}{l}, \inc{\DAG[t_{j+2}]}{l}, \ldots, \inc{\DAG[t_{i-1}]}{l}$ are empty.
If $\inc{\DAG[t_{i-1}]}{l}$ is not empty, then clearly $t_j = t_{i-1}$, i.e., $\slink{t_i}{l} = t_{i-1}$.
Suppose otherwise.
At the previous iteration, the link $\slink{t_{i-1}}{l}$ was determined in such a way that $\DAG\big[\slink{t_{i-1}}{l}\big]$ includes all of $\exc{\DAG[t_{i-1}]}{l}$.
Given that $\inc{\DAG[t_{i-1}]}{l}$ is empty, it follows that $\exc{\DAG[t_{i-1}]}{l}$ is the same as $\exc{\DAG[t_i]}{l}$.
Consequently, $\DAG\big[\slink{t_{i-1}}{l}\big]$ includes all of $\exc{\DAG[t_i]}{l}$, and thus we set $\slink{t_i}{l} = \slink{t_{i-1}}{l}$.

Now, let $u = \ell(t_i)$, $d = d_{\DAG}(u)$, and $\mathbf{t}' = \langle t'_{0}, t'_{1}, \ldots, t'_{d} \rangle$ be the path from the root to the node corresponding to $\DAG_u$.
At line~\ref{line:update-child}, we must set $\clink{t_i}{l}$ to a node $t'_j$ on the path $\mathbf{t}'$, where $j$ is the smallest index such that the $(l\hspace{-0.1em}-\hspace{-0.1em}1)$-cliques in $\DAG[t'_j]$ along with the vertex $u$ form the set $\inc{\DAG[t_i]}{l}$,
or equivalently, $j$ is the smallest integer such that $\DAG[t'_j]$ includes all the $(l\hspace{-0.1em}-\hspace{-0.1em}1)$-cliques in $\DAG_u$.
If $\clink{t'_d}{l\hspace{-0.1em}-\hspace{-0.1em}1}$ is not the root node, it indicates that an $(l\hspace{-0.1em}-\hspace{-0.1em}1)$-clique containing $\ell(t'_d)$ was found inside the recursive call at line~\ref{line:ListingDense}.
In this case, it is clear that $t'_j = t'_d$, i.e., $\clink{t_i}{l} = t'_d$.
Suppose otherwise.
Inside the recursive call, the link $\slink{t'_d}{l\hspace{-0.1em}-\hspace{-0.1em}1}$ was determined in such a way that $\DAG\big[\slink{t'_d}{l\hspace{-0.1em}-\hspace{-0.1em}1}\big]$ includes all of the $(l\hspace{-0.1em}-\hspace{-0.1em}1)$-cliques in $\DAG_u$.
Thus, we set $\clink{t_i}{l} = \slink{t'_d}{l\hspace{-0.1em}-\hspace{-0.1em}1}$.

Finally, at line~\ref{line:mark}, i.e., when we finished listing the sets $\inc{\DAG[t_i]}{l}$ and $\exc{\DAG[t_i]}{l}$
and memoizing them by setting $l$-sibling and $l$-child links, we set $t_i.\textsf{Marked}[l]$ to true.

\begin{algorithm}[t]
    \caption{Listing $k$-cliques using an Induced Subgraph Trie}
    \label{alg:listingIST}
    \DontPrintSemicolon
    \small

    \SetKwBlock{Begin}{function}{end function}

    \SetKwFunction{ListingIST}{\textsf{ListingIST}}

    \SetKw{continue}{continue}
    
    \SetKwInOut{Input}{Input}
    \SetKwInOut{Output}{Output}

    \Input{a node $t$ such that $\DAG[t]$ is \induced{d} or \prefix{d} with respect to a clique $C$ and an integer $l\hspace{-0.1em}=\hspace{-0.1em}k\hspace{-0.1em}-\hspace{-0.1em}|C|$}
    \Output{$k$-cliques in the original graph formed by combining the $l$-cliques in $\DAG[t]$ with the set $C$}
    
    \medskip

    \Begin(\ListingIST{$t$, $l$, $C$}){
        \eIf{$l = 0$}{
            Report $C$ as a $k$-clique.\;
        }{
            \While{\rm $t$ is not the root node}{
                \ListingIST{$\clink{t}{l}$, $l\hspace{-0.1em}-\hspace{-0.1em}1$, $C \cup \{ \ell(t) \}$}\;\label{line:listforchild}
                $t := \slink{t}{l}$\;
            }
        }
    }
\end{algorithm}

The function \textsf{ListingIST} in Algorithm~\ref{alg:listingIST} shows how to retrieve the $l$-cliques within the Induced Subgraph Trie using $l$-child links and $l$-sibling links.
The function call \textsf{ListingIST}($t$, $l$, $C$) finds the same $l$-cliques as \textsf{Listing}($\DAG[t]$, $l$, $C$) does, i.e., it retrieves all the $l$-cliques in the subgraph $\DAG[t]$ from the Induced Subgraph Trie and reports each of these $l$-cliques combined with the set $C$ as a $k$-clique in the original graph.
\textsf{ListingIST} is invoked in \textsf{ListingDense} with respect to a node $t$ such that $t.\textsf{Marked}[l]$ is true, i.e., all the $l$-cliques within the subgraph $\DAG[t]$ are memoized.
When $l$ equals $0$, the function simply reports $C$ as a $k$-clique in the original graph $\DAG$ (lines~2--3).
Otherwise, it first invokes a recursive call to retrieve the set $\inc{\DAG[t]}{l}$, i.e., the $l$-cliques that contain the vertex $\ell(t)$.
Then, it follows the $l$-sibling link $\slink{t}{l}$ to retrieve the set $\exc{\DAG[t]}{l}$, i.e., the $l$-cliques that do not contain the vertex $\ell(t)$.
This process continues until $t$ becomes the root node. 
\begin{theorem}
    Algorithm~2 lists all $k$-cliques in $\DAG$.
\end{theorem}
\begin{proof} 
    Since \DIST differs from the original framework (Algorithm~1) by invoking \textsf{ListingDense} when $\DAG$ satisfies criterion (1), 
    it suffices to show that \textsf{ListingDense}($\DAG, l, C, T$) finds all $l$-cliques in $\DAG$.
    To this end, we show that when \textsf{listingDense} sets $t.\textsf{Marked}[l]$ to true (at line~28), $\clink{t}{l}$ and $\slink{t}{l}$ are created, $\slink{t}{l}$ is $l$-memoized, and $\clink{t}{l}$ is $(l-1)$-memoized.
    At this moment, $\slink{t}{l}$ was created at line~22 and $\clink{t}{l}$ was created at line~27.
    Furthermore, $\textsf{Marked}[l]$ has been set to true for all nodes on the path from the root to the node corresponding to $\DAG[t] \setminus  \ell(t) $, and $\textsf{Marked}[l-1]$ has been set to true for all nodes on the path from the root to the node corresponding to $\DAG[t]_{\ell(t)}$.
    Since $\slink{t}{l}$ is on the path from the root to the node corresponding to $\DAG[t] \setminus \ell(t)$, $\slink{t}{l}.\textsf{Marked}[l]$ is true, and so $\slink{t}{l}$ is $l$-memoized.
    Also, since $\clink{t}{l}$ is on the path from the root to the node corresponding to $\DAG[t]_{\ell(t)}$, $\clink{t}{l}.\textsf{Marked}[l-1]$ is true, and so $\clink{t}{l}$ is $(l-1)$-memoized.
    It follows that $t$ is $l$-memoized.
\end{proof}

%
%

\section{Pruning by Soft Embedding}
\label{sec:pruning}
\begin{figure}[t]
    \centering
    \begin{subfigure}{0.49\linewidth}
        \includegraphics[scale=0.52, trim=5mm 3mm 10mm 3mm, clip=true]{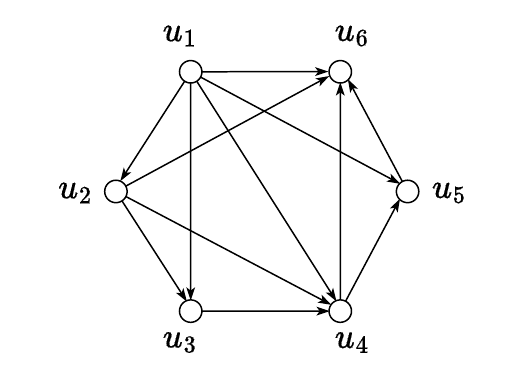}
        \caption{$\DAG$}
        \label{fig:DAG2}
        \end{subfigure}\hfill
    \begin{subfigure}{0.49\linewidth}
        \includegraphics[scale=0.52, trim=10mm 3mm 0mm 3mm, clip=true]{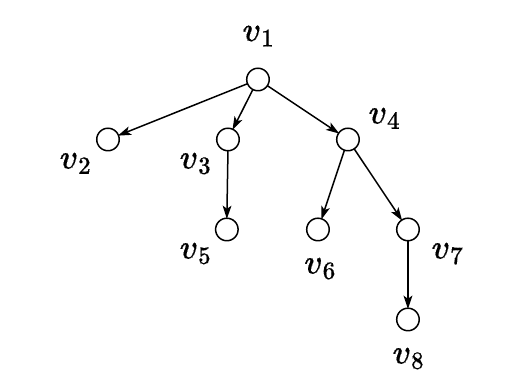}
        \caption{$4$-tree}
        \label{fig:4-tree}
    \end{subfigure}

    \begin{subfigure}{0.98\linewidth}
    \centering
        \hspace*{-3mm}
        \includegraphics[scale=0.52, trim=9.0mm 0.0mm 9.5mm 0mm, clip=true]{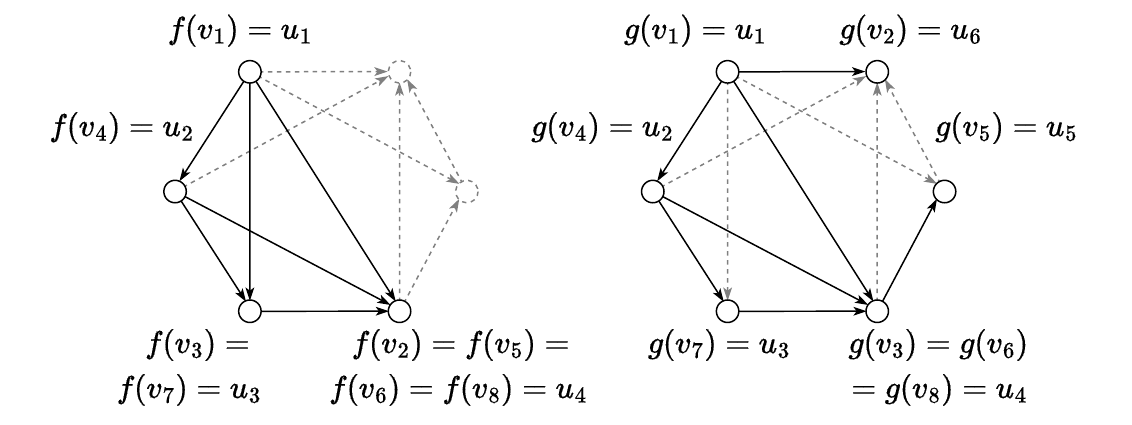}
        \begin{minipage}[t]{.56\linewidth}
            \centering
            \subcaption[]{Soft embedding of $4$-tree \\ \hphantom{(c)} that is a $4$-clique}
            \label{fig:tree-embedding1}
        \end{minipage}
        \begin{minipage}[t]{.42\linewidth}
            \centering
            \subcaption[]{Soft embedding of $4$-tree \\ \hphantom{(d)} that is not a $4$-clique}
            \label{fig:tree-embedding2}
        \end{minipage}
    \end{subfigure}
    \caption{DAG $\DAG$, $4$-tree, and two soft embeddings $f$ and $g$.}
\end{figure}
\noindent
In this section, we introduce a novel pruning technique designed to efficiently reduce the search space.
Li et al. \cite{OrderingHeuristics} used a pruning technique based on greedy coloring.
In their approach, within each $1$-induced subgraph, greedy coloring is applied and the vertices are reordered according to their color values so that the color values can be used to reduce the search space.
However, this reordering of the vertices within induced subgraphs in Li et al.'s algorithm makes it incompatible with the Induced Subgraph Trie,
where all stored subgraphs must follow a common total ordering.
Unlike the approach by Li et al.,
our method does not require reordering vertices within induced subgraphs, and it is fully compatible with the Induced Subgraph Trie.
\begin{definition}
   An \emph{$l$-tree} is a directed rooted tree such that
   \begin{enumerate}
       \item if $l=1$, then the $l$-tree is a single-vertex tree, and
       \item otherwise, the root of the $l$-tree has $l\hspace{-0.1em}-\hspace{-0.1em}1$ child nodes, where the subtree rooted at the $i$-th child is an $i$-tree.
   \end{enumerate}
\end{definition}
\begin{example}
Figure~\ref{fig:4-tree} shows the $4$-tree.
The root $v_1$ has $3$ child nodes, namely, $v_2$, $v_3$, $v_4$.
The subtrees rooted at $v_2$, $v_3$, and $v_4$ are a $1$-tree, a $2$-tree, and a $3$-tree, respectively.
\end{example}
\begin{definition}
    A \emph{soft embedding} of an $l$-tree $T$ in a DAG $\DAG$ is a mapping $f: V_T \to V_{\DAG}$ such that
    \begin{enumerate}
        \item if $(v, v') \in E_{T}$, then $(f(v), f(v')) \in E_{\DAG}$, and
        \item for sibling nodes $v$ and $v'$ such that $d_{T}(v) < d_{T}(v')$, it holds that $f(v) \succ f(v')$.
    \end{enumerate}
\end{definition}
Notice that a soft embedding of an $l$-tree does not need to be injective.
If a soft embedding of an $l$-tree in $\DAG$ maps the root of the $l$-tree to a vertex $u \in V_{\DAG}$, the soft embedding is said to be \emph{rooted at} $u$.
\begin{example}\label{example:embedding}
    Consider the DAG $\DAG$ in Figure~\ref{fig:DAG2}. Figures \ref{fig:tree-embedding1} and~\ref{fig:tree-embedding2} show two soft embeddings of $4$-tree in $\DAG$ rooted at $u_1$, i.e.,
    $f\hspace{-0.3em}=\hspace{-0.3em}\big\{ (v_1, u_1), (v_2, u_4), (v_3, u_3), (v_4, u_2), (v_5, u_4), (v_6, u_4),\allowbreak (v_7, u_3), (v_8, u_4) \big\}$\hspace{0.1em}and\hspace{0.1em}$g\hspace{-0.2em}=\hspace{-0.3em}\big\{ (v_1, u_1), (v_2, u_6), (v_3, u_4), (v_4, u_2),\allowbreak (v_5, u_5), (v_6, u_4), (v_7, u_3), (v_8, u_4) \big\}$.
    Note that the image of the $4$-tree through $f$, i.e., $\{ u_1, u_2, \allowbreak u_3, u_4 \}$, is a $4$-clique in $\DAG$.
\end{example}
An $l$-tree can be viewed as an $l$-clique torn apart to form a tree;
given an $l$-tree, if we contract the vertices with the same out-degree into a single vertex, it becomes an $l$-clique.

\begin{lemma}
    Given a DAG $\DAG$, if there exists no embedding of an $l$-tree rooted at a vertex $u \in V_{\DAG}$, then there exists no $l$-clique in $\DAG$ with $u$ as its minimum vertex.
\end{lemma}
\begin{proof}
    If $\DAG$ includes an $l$-clique with $u$ as its minimum vertex, then there exists a soft embedding of an $l$-tree that maps the root of the $l$-tree to the minimum vertex $u$ of this $l$-clique, as shown in Figure~\ref{fig:tree-embedding1}.
\end{proof}

We focus on identifying the maximum $l$ value for each vertex $u \in V_{\DAG}$ such that a soft embedding of an $l$-tree in $\DAG$ rooted at $u$ exists, denoted by $r_{\DAG}(u)$.
This is because if there exists a soft embedding of an $l$-tree in $\DAG$ rooted at $u$,
then there also exists a soft embedding of an $(l\hspace{-0.1em}-\hspace{-0.1em}1)$-tree in $\DAG$ rooted at $u$.

Algorithm~\ref{alg:pruning} computes $r_{\DAG}(u)$ for each vertex $u$ of $\DAG$, following descending order in the total ordering of vertices in $\DAG$.
Assume we are to compute $r_{\DAG}(u)$ for a vertex $u$, and $u$ has $d$ out-neighbors $u_1, u_2, \ldots, u_{d}$, ordered such that $u_1\hspace{-0.2em}\prec\hspace{-0.2em}u_2\hspace{-0.2em}\prec\hspace{-0.2em}\cdots\hspace{-0.2em}\prec\hspace{-0.2em}u_d$.
Initially, we set $r_{\DAG}(u)\hspace{-0.2em}:=\hspace{-0.2em}1$.
Then, we iterate through the out-neighbors of $u$ in descending order, starting with $u_{d}$ and progressing to $u_1$. 
If $r_{\DAG}(u) \le r_{\DAG}(u_i)$,
we increment $r_{\DAG}(u)$ by one;
otherwise, we proceed to the next out-neighbor of $u$.
As Algorithm~5 considers each edge of $\DAG$ exactly once, it takes $\mathcal{O}(\EG$) time.

\begin{theorem}
    Given a DAG $\DAG$, for each vertex $u \in \DAG$, $r_{\DAG}(u)$ computed by Algorithm~5 is the maximum $l$ value such that a soft embedding of an $l$-tree rooted at $u$ exists.
\end{theorem}
\begin{proof}
We begin by observing the following property of $l$-trees.
Let $T$ be an $l$-tree with root $v$, and $v'$ be the last child node of $v$, implying that the subtree of $T$ rooted at $v'$ is an $(l\hspace{-0.1em}-\hspace{-0.1em}1\hspace{-0.1em})$-tree.
By removing the directed edge between $v$ and $v'$, $T$ can be split into two separate $(l\hspace{-0.1em}-\hspace{-0.2em}1\hspace{-0.1em})$-trees.
The first tree, rooted at $v'$, remains an unchanged $(l\hspace{-0.1em}-\hspace{-0.2em}1\hspace{-0.1em})$-tree.
The second tree, rooted at $v$ now with $l\hspace{-0.1em}-\hspace{-0.2em}2\hspace{-0.1em}$ child nodes,
is another $(l\hspace{-0.1em}-\hspace{-0.2em}1\hspace{-0.1em})$-tree. \\
This observation gives a method for determining the existence of a soft embedding of an $l$-tree rooted at a vertex $u$ in $\DAG$.
Suppose we have an embedding $g_1$ of an $(l\hspace{-0.1em}-\hspace{-0.2em}1\hspace{-0.1em})$-tree rooted at $u$, with $g_1$ mapping the last child node of the root to an out-neighbor $u'$ of $u$.
Further assume we have another soft embedding $g_2$ of an $(l\hspace{-0.1em}-\hspace{-0.2em}1\hspace{-0.1em})$-tree rooted at $u''$, where $u''$ is an out-neighbor of $u$ such that $u'' \prec u'$.
Then, combining $g_1$ and $g_2$ gives a soft embedding of an $l$-tree rooted at $u$. \\
Now consider the computation of $r_{\DAG}(u)$.
It is important to note that $r_{\DAG}(u')$ for each out-neighbor $u'$ was computed previously.
Algorithm~5 initializes $r_{\DAG}(u)$ to 1, because $u$, as a single vertex, is a $1$-tree.
At each iteration of the for-loop at line~4, Algorithm~5 increments $r_{\DAG}(u)$ by one if and only if $r_{\DAG}(u) \le r_{\DAG}(u')$, where $u'$ is an out-neighbor of $u$.
At this moment, $r_{\DAG}(u) \le r_{\DAG}(u')$ indicates the presence of two distinct soft embeddings of an $r_{\DAG}(u)$-tree, respectively rooted at $u$ and $u'$.
These can be combined to form a soft embedding of an $(r_{\DAG}(u)\hspace{-0.05em}+\hspace{-0.05em}1)$-tree rooted at $u$.
When the for-loop at line~4 finishes, $r_{\DAG}(u)$ is the maximum $l$ value such that a soft embedding of an $l$-tree rooted at $u$ exists.
\end{proof}

In line 10 of Algorithm~\ref{alg:listing} and line 23 of Algorithm~\ref{alg:listingDense}, we apply Algorithm~\ref{alg:pruning} on $\DAG' = \DAG\big[ N(u) \cup \{ u \} \big]$ and obtain $r_{\DAG'}(u)$.
If $r_{\DAG'}(u) < l$ after Algorithm~\ref{alg:pruning} is run on $\DAG'$, then there exists no soft embedding of an $l$-tree rooted at $u$.
By Lemma~1, we can safely prune the search space to find $l$-cliques containing $u$ as the minimum vertex.

\begin{algorithm}[t]
    \caption{Computing $r_{\DAG}(u)$\hspace{0.2em}for\hspace{0.2em}every\hspace{0.2em}vertex\hspace{0.2em}$u$\hspace{0.2em}in\hspace{0.2em}$\DAG$}
    \label{alg:pruning}
    \DontPrintSemicolon
    \small

    \SetKwBlock{Begin}{function}{end function}

    \SetKwFunction{Pruning}{\textsf{ComputeSoftEmbeddings}}

    \SetKw{continue}{continue}
    
    \SetKwInOut{Input}{Input}
    \SetKwInOut{Output}{Output}

        \ForEach{\rm $u \in V_{\DAG}$ in descending order}{
            $r_{\DAG}(u) := 1$\;
            Let $u_1, u_2, \ldots, u_d$ be the out-neighbors of $u$ such that $u_1 \prec u_2 \prec \cdots \prec u_d$.\;
            \For{\rm\sffamily $i := d$ to $1$}{
                \If{\rm\sffamily $r_{\DAG}(u) \le r_{\DAG}(u_i)$}{
                    Increment $r_{\DAG}(u)$ by one.\;
                }
            }
        }
\end{algorithm}

%

\section{Space Usage of \DIST}
\noindent
In this section, we analyze the space usage of \DIST and introduce a technique to manage the space usage.

Since a $k$-clique is included in the $1$-hop neighborhood of its minimum vertex, we may build one Induced Subgraph Trie for each subtask at depth 1 of the recursion, i.e.,
finding $(k\hspace{-0.1em}-\hspace{-0.1em}1)$-cliques within a $1$-induced subgraph.
Our empirical study has shown the following results (see Figure~\ref{fig:graph-vs-trie}).
\begin{itemize}[leftmargin=*]
    \item For large $k$ values, the space usage of \DIST is negligible since few $1$-induced subgraphs require building the Induced Subgraph Trie.
    This is because most $1$-induced subgraphs are pruned due to the pruning technique.
    Furthermore, for large $k$ values, subgraphs are less likely to satisfy criterion (1), so fewer subgraphs get stored in the Induced Subgraph Trie, leading to smaller space requirements for \DIST.
    \item For small $k$ values, although many $1$-induced subgraphs pass the filtering and requires building an Induced Subgraph Trie, the number of nodes in the Induced Subgraph Trie grows almost linearly with respect to the size of the $1$-induced subgraph and rarely requires excessive space.
\end{itemize}
These observations allow \DIST to share the Induced Subgraph Trie across multiple $1$-induced subgraphs, 
which speeds up the listing process, because subtasks at depth 1 are not entirely independent.

Now we introduce a bounding technique to efficiently manage the space usage of \DIST using a soft bound $\tau$ and a hard bound $2\tau$.
Since each node in the Induced Subgraph Trie has less than $k$ $l$-sibling/$l$-child links, the Induced Subgraph Trie uses $\mathcal{O}(kN)$ space, where $N$ is the number of nodes in the Induced Subgraph Trie.
After completing each subtask at depth 1, \DIST checks whether $kN$ is less than or equal to the soft bound $\tau$.
If so, \DIST continues using the same Induced Subgraph Trie for the next subtask at depth 1;
otherwise, \DIST deletes the Induced Subgraph Trie and starts with an empty one.
To prevent extreme cases where the Induced Subgraph Trie might grow excessively large,
\DIST deletes the Induced Subgraph Trie whenever $kN$ exceeds the hard bound $2\tau$, and it starts with an empty one.
This way, \DIST uses $\mathcal{O}(\tau)$ space.
In practice, we find $\tau = 16$ million sufficient to demonstrate the efficiency of \DIST while limiting memory usage to only a few GBs.

\begin{figure*}[h]
    \includegraphics[scale=0.26, trim=0mm 0mm 0mm 12mm]{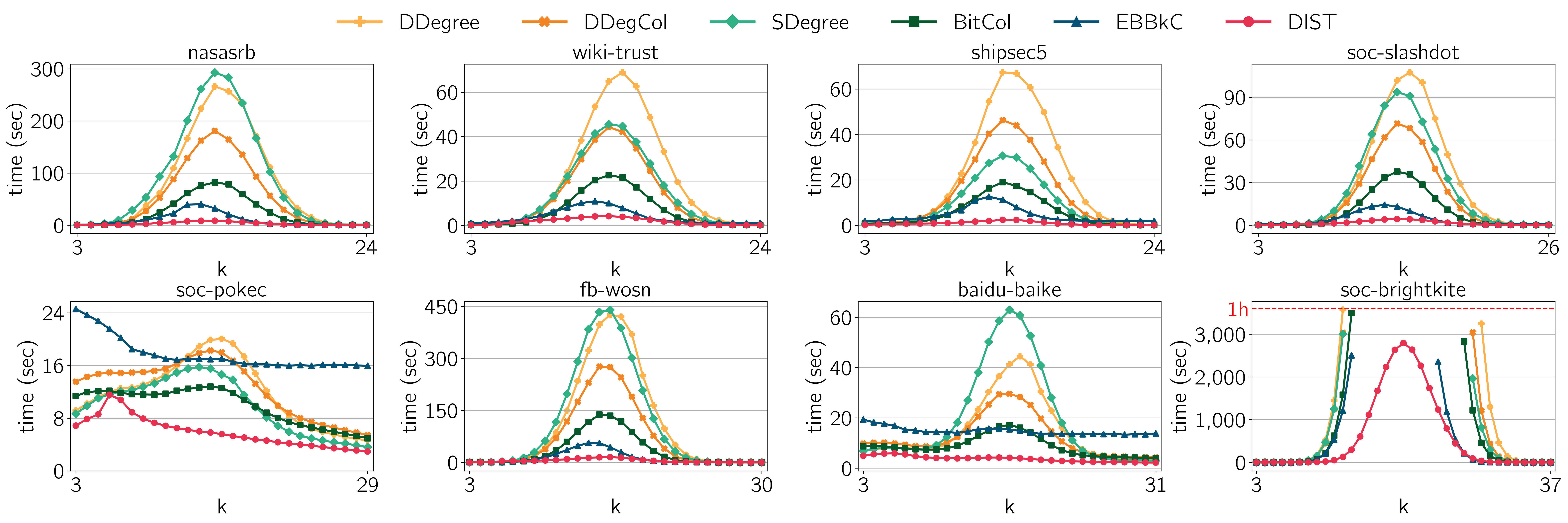}
    \caption{Running time of algorithms on small-$\omega$ graphs.}
    \label{fig:exp-small}
\end{figure*}
\begin{figure*}[h]
    \includegraphics[scale=0.26, trim=0mm 0mm 0mm 12mm]{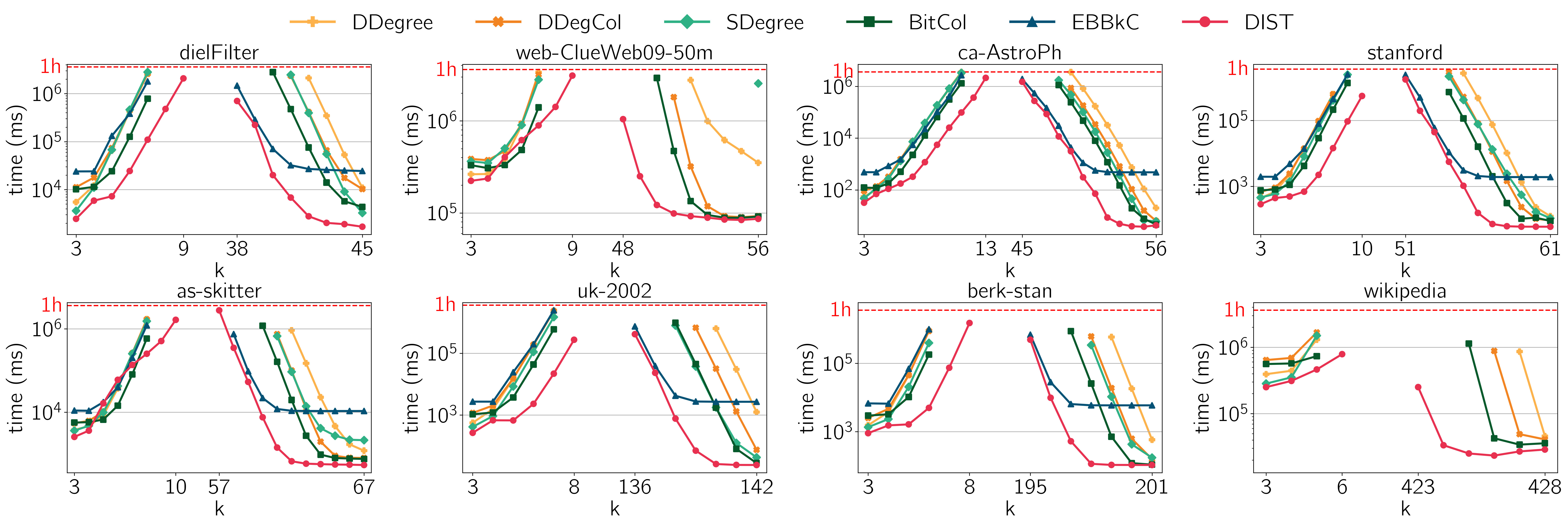}
    \caption{Running time of algorithms on large-$\omega$ graphs. The $y$-axis is in a logarithmic scale.}
    \label{fig:exp-large}
\end{figure*}

\section{Performance Evaluation}
\noindent
In this section, we conduct experiments to demonstrate the efficiency of our algorithm \DIST.

\begin{table}[t]
    \caption{Datasets listed by the size $\omega(G)$ of maximum clique }
    \label{tab:datasets}
    \centering
    \begin{tabular}{c|cccc}
        \toprule
        Dataset & $|V_G|$ & $|E_G|$ & $\delta(G)$ & $\omega(G)$ \\
        \midrule
        \textsf{nasasrb}            & 54,870       & 1,311,227     & 35    & 24 \\
        \textsf{wiki-trust}         & 138,587      & 715,883       & 55    & 24 \\
        \textsf{shipsec5}           & 179,104      & 2,200,076     & 29    & 24 \\
        \textsf{soc-slashdot}       & 70,068       & 358,647       & 53    & 26 \\
        \textsf{soc-pokec}          & 1,632,803    & 22,301,964    & 47    & 29 \\ 
        \textsf{fb-wosn}            & 63,731       & 817,090       & 52    & 30 \\
        \textsf{baidu-baike}        & 2,141,300    & 17,014,946    & 78    & 31 \\
        \textsf{soc-brightkite}     & 56,739       & 212,945       & 52    & 37 \\
        \midrule     
        \textsf{dielFilter}         & 420,409      & 16,653,308    & 58   & 45\\
        \textsf{web-ClueWeb09-50m}  & 428,136,613  & 454,472,685   & 193   & 56 \\
        \textsf{ca-AstroPh}         & 17,903       & 196,972       & 57    & 57 \\
        \textsf{stanford}           & 281,903      & 1,992,636     & 71    & 61 \\
        \textsf{as-skitter}         & 1,696,415    & 11,095,298    & 111   & 67 \\
        \textsf{uk-2002}            & 191,105      & 2,192,873     & 142   & 142 \\
        \textsf{berk-stan}          & 685,240      & 6,649,470     & 201   & 201 \\
        \textsf{wikipedia}          & 25,921,548   & 543,183,611   & 1,120  & 428 \\ 
        \bottomrule
    \end{tabular}
\end{table}

\subsection{Experimental Setup}
\noindent
\textbf{Algorithms.} The compared algorithms are as follows.
\begin{itemize}[leftmargin=*, labelsep=0.10em]
    \item \DDegree: a degree-ordering based algorithm~\cite{OrderingHeuristics}
    \item \DDegCol: a color-ordering based algorithm \cite{OrderingHeuristics}
    \item \SDegree:\hspace{0.2em}state-of-the-art\hspace{0.2em}degree-ordering\hspace{0.2em}based\hspace{0.2em}algorithm\hspace{0.15em}\cite{SetIntersectionSpeedup}
    \item \BitCol: state-of-the-art color-ordering based algorithm \cite{SetIntersectionSpeedup}
    \item \EBBkC:\hspace{0.05em}state-of-the-art\hspace{0.16em}edge-oriented\hspace{0.16em}branching\hspace{0.16em}algorithm\hspace{0.05em}\cite{EBBkC}
    \item \DIST: our algorithm
\end{itemize}
The source codes of \DDegree/\DDegCol\footnote{https://github.com/Gawssin/kCliqueListing}, \SDegree/\BitCol\footnote{https://github.com/zer0y/k-clique-listing} and \EBBkC\footnote{https://github.com/wangkaixin219/EBBkC} are publicly available.
In the experiments, we use bitmask width $\mathcal{L} = 24$ for \BitCol, as suggested in \cite{SetIntersectionSpeedup}.
Also, \EBBkC used $2$-plex early termination for $k \le \tau/2$ and $3$-plex early termination for $k > \tau/2$, as suggested in \EBBkC.
For large $k$ values in large-$\omega$ datasets, \DIST also used the $2$-plex early termination technique.

\noindent
\textbf{Datasets.} We obtained 16 real-world graphs from Network Repository\footnote{https://networkrepository.com/}.
Table~\ref{tab:datasets} lists the datasets used in the experiments.
The datasets are classified into two categories, depending on the size $\omega$ of the maximum clique.
The upper section of Table~\ref{tab:datasets} contains small-$\omega$ graphs,
and the lower section contains large-$\omega$ graphs.

\noindent
\textbf{Experimental settings.} The experiments were conducted on a Linux machine equipped with two Intel Xeon CPUs (E5-2680 v3 @ 2.50GHz).
Reported running times include both preprocessing and listing $k$-cliques.
We set a time limit of 1 hour for every experiment.
When we compare two algorithms, we compare the total running time that an algorithm takes for all of the $k$ values for which both of the algorithms finished within the time limit.
The overall speedup of our algorithm is given by the geometric mean of the ratio of total running times, taken over all datasets.

\subsection{Single-Threaded}\label{sec:exp-main}
\noindent
\textbf{Experiment on small-$\omega$ graphs.} Figure~\ref{fig:exp-small} shows the results for small-$\omega$ graphs using a single thread.
For each dataset, we measured the running time of each algorithm varying $k$ from $3$ to the maximum clique size, $\omega(G)$.
Our algorithm \DIST is the best performing algorithm among the competitors for most of the instances, followed by \EBBkC and \BitCol.
Specifically, for the \textsf{soc-brightkite} dataset, \DIST is the only algorithm capable of solving the problem within the given time limit for $k$ values between 15 and 23.
In general, the mid-range $k$ values (i.e., the $k$ values around $\omega(G) / 2$) are the most challenging ones for the $k$-clique listing problem,
and the running times of the algorithms tend to peak at mid-range $k$ values.
Our algorithm \DIST outperforms \BitCol with a factor of 5.06 in overall speedup and
up to a factor of 13.19 (when $k=28$ in \textsf{soc-brightkite}),
and it outperforms \EBBkC with a factor of 3.13 in overall speedup and
up to a factor of 22.86 (when $k=24$ in \textsf{shipsec5}, \EBBkC takes 1909.54 ms and \DIST takes 83.50 ms.
Also, when $k=14$ in \textsf{soc-brightkite}, \EBBkC takes 2504.01 sec and \DIST 296.63 sec.)
Since \DIST memoizes the cliques upon their initial computation, the running time of \DIST stays relatively low compared to the other algorithms even when the $k$ value reaches the challenging range.
Notably, for $k=28$ in \textsf{soc-brightkite},
\DIST made around 55 billion recursive calls to \textsf{ListingIST} but less than 4 million recursive calls to \textsf{Listing} and \textsf{ListingDense} combined
(whereas \BitCol made around 55 billion recursive calls for computing cliques),
implying that more than 99.99\% of the tasks for computing cliques are replaced by retrieving the cliques from the Induced Subgraph Trie, rather than recomputing.
As a result, \DIST achieved a remarkable speedup over \BitCol for this instance.

\noindent
\textbf{Experiment on large-$\omega$ graphs.} Figure~\ref{fig:exp-large} shows the results for large-$\omega$ graphs using a single thread.
Due to the exponential increase in the number of cliques, conducting experiments for mid-range $k$ values is infeasible.
For each dataset, we evaluated the running times of the algorithms by incrementally increasing $k$ from $3$ and decreasing $k$ from the maximum clique size, $\omega(G)$, until none of the algorithms could solve the problem within the given time limit.
Our algorithm \DIST is the only algorithm capable of solving the problem within the given time limit for a challenging range of $k$ values, i.e., for $k$ values nearing $\omega(G)/2$.
It significantly outperforms \BitCol with a factor of 14.47 in overall speedup
and up to a factor of 1697.23 (when $k=197$ in \textsf{berk-stan}),
and it outperforms \EBBkC with a factor of 3.66 in overall speedup and up to a factor of 201.32 (when $k=6$ in \textsf{berk-stan}).

\subsection{Evaluation of Techniques}
\noindent
In this section, we evaluate the effect of our techniques in terms of running time and memory usage.
We compare (1) \textsf{baseline}, in which none of our techniques are used (i.e., Algorithm~\ref{alg:framework}), (2) \textsf{memo}, in which only the memoization technique using the Induced Subgraph Trie is used, and (3) \textsf{memo+pruning}, in which both the memoization and the pruning techniques are used (i.e., \DIST).
We measure the memory usage of \textsf{baseline}, \textsf{memo}, and \textsf{memo+pruning} (\DIST)
along with \EBBkC and \BitCol, two top-performing runner-ups of the experiment in Section~\ref{sec:exp-main}.

\begin{figure*}[t]
    \centering
    \includegraphics[scale=0.26, trim=0mm 0mm 0mm 12mm]{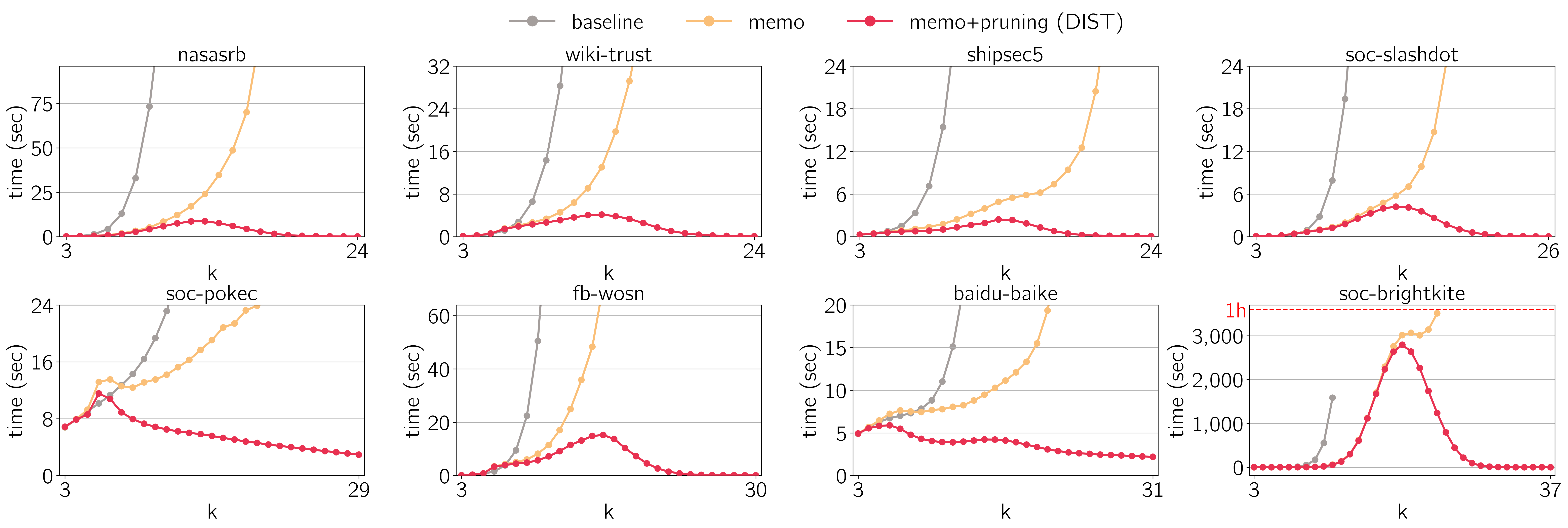}
    \caption{Effect of our techniques on running time.}
    \label{fig:exp-abl-time}
\end{figure*}
\begin{figure*}[t]
    \includegraphics[scale=0.26, trim=0mm 0mm 0mm 12mm]{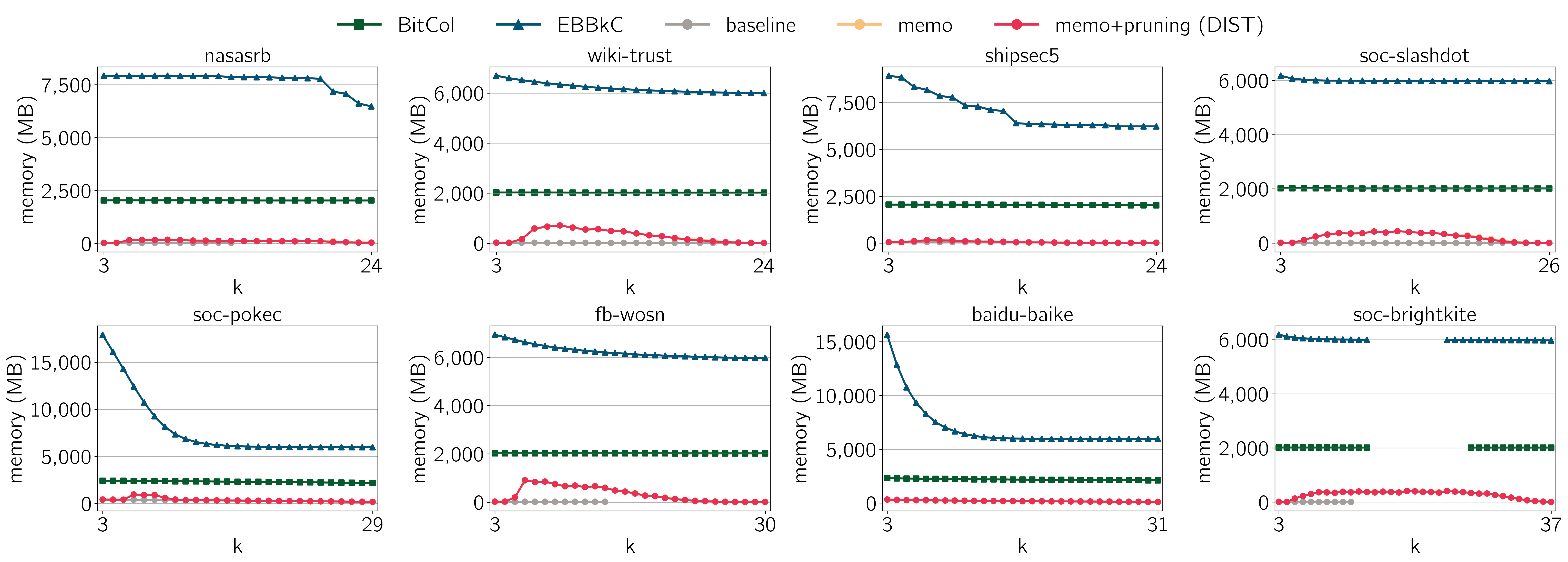}
    \caption{Peak memory usage of compared algorithms.}
    \label{fig:exp-abl-space}
\end{figure*}

\begin{figure}
    \centering
    \includegraphics[width=\linewidth]{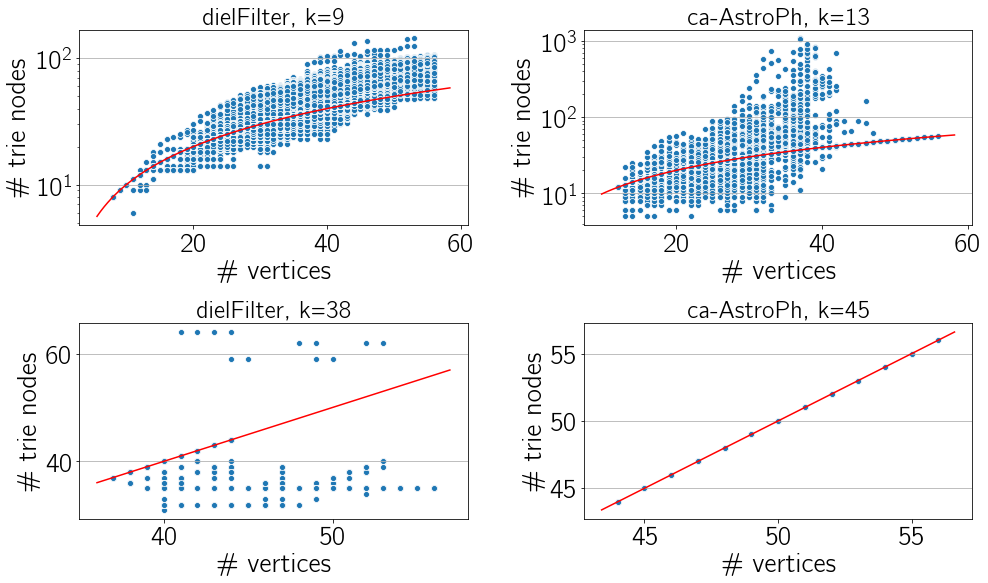}
    \caption{Size of Induced Subgraph Trie for a $1$-induced subgraph. Each data point represents a $1$-induced subgraph, where the $x$-axis value is the number of vertices in the $1$-induced subgraph and the $y$-axis value is the number of nodes in the Induced Subgraph Trie built for the $1$-induced subgraph. The red line represents the line $y=x$, which serves as a reference to show the general trend.}
    \label{fig:graph-vs-trie}
\end{figure}

\noindent
\textbf{Running time.} Figure~\ref{fig:exp-abl-time} shows the running times of \textsf{baseline}, \textsf{memo}, and \textsf{memo+pruning} (\DIST).
In general, \textsf{memo+pruning}\hspace{0.2em}(\DIST) is the best performer among the three variants.
The memoization technique using the Induced Subgraph Trie significantly boosts the listing process when compared to \textsf{baseline},
although this advantage diminishes for very small $k$ values due to overhead associated with constructing the Induced Subgraph Trie.
As $k$ gets larger, the difference between the running time of \textsf{baseline} and that of \textsf{memo} becomes more evident.
This is because as $k$ gets larger, the listing framework creates an exponentially growing number of recursive calls, or subtasks.
While \textsf{baseline} handles each of these subtasks from scratch, solving the same subtask repeatedly, \textsf{memo} solves the same subtask only once.
The pruning technique further improves performance, particularly at larger $k$ values.
This is because a greater number of subgraphs are subject to pruning, due to the fact that it is less likely for a subgraph to contain a soft embedding of an $l$-tree for larger $l$ values.

\begin{figure*}[t]
    \includegraphics[scale=0.26, trim=0mm 0mm 0mm 12mm]{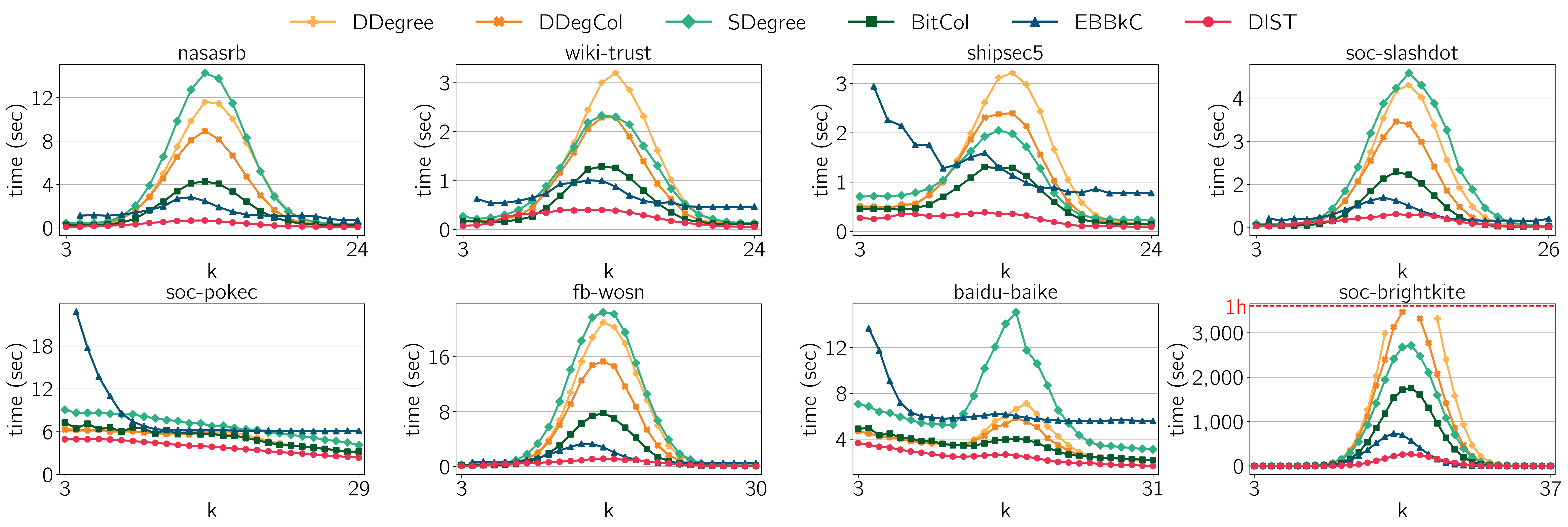}
    \caption{Running time of parallel algorithms on small-$\omega$ graphs using 24 threads.}
    \label{fig:exp-small-parallel}
\end{figure*}
\begin{figure*}[t]
    \includegraphics[scale=0.26, trim=0mm 0mm 0mm 12mm]{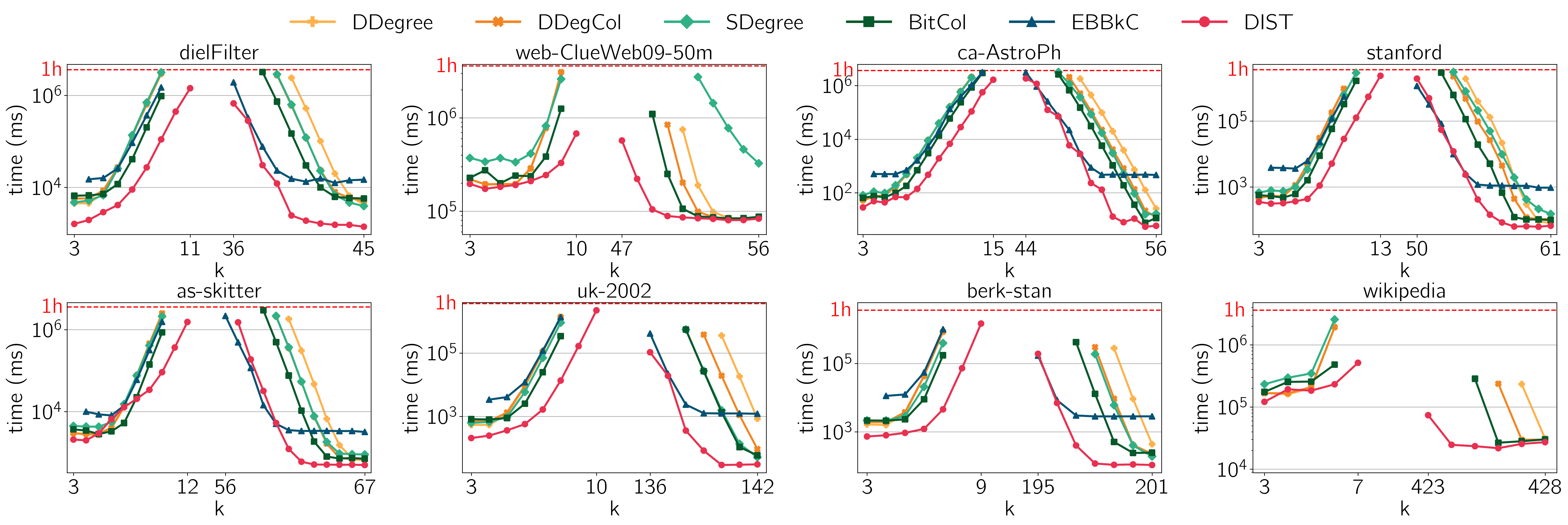}
    \caption{Running time of parallel algorithms on large-$\omega$ graphs using 24 threads.
    The $y$-axis is in a logarithmic scale.}
    \label{fig:exp-large-parallel}
\end{figure*}

\noindent
\textbf{Space usage.}
Figure~\ref{fig:graph-vs-trie} shows the size of Induced Subgraph Trie for a $1$-induced subgraph, for {\small\textsf{dielFilter}} and {\small\textsf{ca-AstroPh}}, which serve as examples illustrating the general trend.
We omit any $1$-induced subgraph if it is pruned or if the resulting Induced Subgraph Trie is empty.
As discussed in Section~VI, the Induced Subgraph Trie for a single $1$-induced subgraph rarely requires excessive space, allowing multiple $1$-induced subgraphs to share the same Induced Subgraph Trie. 

Figure~\ref{fig:exp-abl-space} shows the peak memory usage of the compared algorithms, where
all points of \textsf{memo} overlap with those of \textsf{memo+\allowbreak pruning} (\DIST). The reason is that although the pruning technique effectively removes subgraphs that include no $l$-clique, the number of subgraphs that need to be memoized remains unchanged, as \DIST inserts an $l$-induced subgraph into the Induced Subgraph Trie only when it is guaranteed to include an $l$-clique.
\BitCol, which uses an auxiliary data structure to speed up the computation of neighborhood intersections,
consistently requires at least 2GB of memory for any instance.
\EBBkC uses at least 5GB of memory for maintaining auxiliary data structures for the edge-oriented branching,
and requires more memory especially for larger graphs (e.g., \textsf{soc-pokec} and \textsf{baidu-baike}).
In contrast, \textsf{memo} and \textsf{memo+pruning} (\DIST) use less than 2GB of memory in all instances due to the bounding technique explained in Section~VI.
As the value of $k$ increases, the memory usage of our algorithm first increases 
because more subgraphs are stored into the Induced Subgraph Trie, and then
it decreases because subgraphs are less likely to meet the density threshold derived from Tur{\'a}n's theorem and thus less likely to be stored into the Induced Subgraph Trie.

\subsection{Parallelization}
\noindent
Danisch et al. \cite{kClist} proposed two strategies for parallelizing the $k$-clique listing algorithms, namely \textsf{NodeParallel} and \textsf{EdgeParallel}.
In the context of listing $k$-cliques in $\DAG$,
the \textsf{NodeParallel} strategy divides the task into $\hspace{-0.1em}\VG\hspace{-0.1em}$ subtasks, with each subtask listing $(k\hspace{-0.2em}-\hspace{-0.2em}1)$-cliques within $\DAG\big[ N(u) \big]$ for a vertex $u \in V_{\DAG}$.
On the other hand, the \textsf{EdgeParallel} strategy divides the task into $\hspace{-0.1em}\EG\hspace{-0.1em}$ subtasks, with each subtask listing $(k\hspace{-0.2em}-\hspace{-0.2em}2)$-cliques within $\DAG\big[ N(u) \cap N(v) \big]$ for an edge $(u, v) \in E_{\DAG}$.
These subtasks are then processed concurrently across multiple threads.
Prior research \cite{kClist, OrderingHeuristics, SetIntersectionSpeedup} has demonstrated that, for existing $k$-clique listing algorithms, the \textsf{EdgeParallel} strategy generally outperforms the \textsf{NodeParallel} strategy,
since it breaks down the whole task into smaller subtasks, thus enhancing parallelism.

However, our empirical study has shown a different result for our algorithm \DIST.
Specifically, \DIST shows improved performance under the \textsf{NodeParallel} strategy.
This improved performance with the \textsf{NodeParallel} strategy is due to the memoization technique,
which gains efficiency when subtasks that are locally proximate are executed by the same thread.
Thus, in this section, we compare \DIST under the \textsf{NodeParallel} strategy to existing algorithms under the \textsf{EdgeParallel} strategy.

\noindent
\textbf{Experiment on small-$\omega$ graphs.} Figure~\ref{fig:exp-small-parallel} shows the results for small-$\omega$ graphs using 24 threads.
For each dataset, we measured the running time of each algorithm varying $k$ from $3$ to the maximum clique size, $\omega(G)$.
The observed trend aligns with that from single-thread experiments, with the running times generally peaking at mid-range $k$ values.
The running time of our algorithm \DIST stays relatively low, being the best performing algorithm among the competitors.
Notably, \DIST outperformed competing algorithms across all $k$ values for \textsf{nasasrb}, \textsf{shipsec5}, \textsf{soc-pokec}, and \textsf{baidu-baike} datasets.
For other datasets, the superiority of \DIST is most evident at mid-range $k$ values, which are typically the most challenging.
\DIST outperforms \BitCol with a factor of 2.79 in overall speedup
and up to a factor of 9.57 (when $k=13$ in \textsf{soc-brightkite}),
and it outperforms \EBBkC with a factor of 2.77 in overall speedup and up to a factor of 18.56 (when $k=37$ in \textsf{soc-brightkite}, \EBBkC takes 296.85 ms and \DIST takes 15.99 ms.
Also, when $k=14$ in \textsf{soc-brightkite}, \EBBkC takes 121.38 sec and \DIST 16.60 sec.)

\noindent
\textbf{Experiment on large-$\omega$ graphs.} Figure~\ref{fig:exp-large-parallel} shows the results for large-$\omega$ graphs using 24 threads.
For each dataset, we evaluated the running times of the algorithms by incrementally increasing $k$ from $3$ and decreasing $k$ from the maximum clique size, $\omega(G)$, until none of the algorithms could solve the problem within the given time limit.
Notably, \DIST outperformed competing algorithms across all $k$ values for {\small\textsf{dielFilter}}, {\small\textsf{web-ClueWeb09-50m}}, {\small\textsf{uk-2002}}, and {\small\textsf{wikipedia}} datasets,
being the only algorithm capable of solving the problem within the given time limit for a challenging range of $k$-values.
It significantly outperforms \BitCol with a factor of 15.35 in overall speedup
and up to a factor of 1534.72 (when $k=138$ in \textsf{uk-2002}),
and it outperforms \EBBkC with a factor of 3.11 in overall speedup and up to a factor of 218.77 (when $k=7$ in \textsf{berk-stan}).

\section{Conclusion}
\noindent
In this paper, we proposed an innovative data structure called the \emph{Induced Subgraph Trie}, which enables us to store cliques of a graph compactly and retrieve them efficiently.
Using the Induced Subgraph Trie, we designed an algorithm \DIST for the $k$-clique listing problem, 
which outperforms the state-of-the-art algorithm in both time and space usage.
It will be an interesting future work to find more applications of the Induced Subgraph Trie.

\balance
\bibliographystyle{IEEEtran}
\bibliography{references}

\end{document}